\pdfoutput=1
\documentclass[journal,10pt]{IEEEtran}
\usepackage{cite}
\usepackage{indentfirst}
\usepackage{graphicx}
\usepackage{changepage}
\usepackage{stfloats}
\usepackage{amsfonts,amssymb}
\usepackage{bm}
\usepackage{algorithm}
\usepackage{algorithmic}
\usepackage{amsmath}
\usepackage{amsmath,amssymb,amsfonts}
\usepackage{times}
\usepackage{url}
\usepackage{cite}
\usepackage{textcomp}
\usepackage{xcolor}
\usepackage{color}
\usepackage{setspace}
\usepackage{enumitem}
\usepackage{float}
\usepackage{diagbox}
\usepackage[T1]{fontenc}
\usepackage[utf8]{inputenc}
\usepackage{authblk}
\usepackage{threeparttable}
\usepackage{booktabs}
\usepackage{footnote}
\usepackage{graphicx}
\usepackage{pifont}
\usepackage{subfigure}
\usepackage{mathrsfs}
\usepackage{bbm}
\usepackage{dsfont}

\allowdisplaybreaks[4]
\newenvironment{proof}{{\indent  \indent \it Proof:}}{\hfill $\blacksquare$}

\begin{document}
\title{Cooperative ISAC Networks: Performance Analysis, Scaling Laws and Optimization}

\author{
	Kaitao Meng, \textit{Member, IEEE}, Christos Masouros \textit{Fellow, IEEE}, Athina P. Petropulu, \textit{Fellow, IEEE}, and Lajos Hanzo, \textit{Life Fellow, IEEE}
	\thanks{Part of this work was submitted to the Conference IEEE SPAWC 2024  \cite{meng2024cooperative} under review.}
	\thanks{K. Meng and C. Masouros are with the Department of Electronic and Electrical Engineering, University College London, London WC1E 7JE, UK (emails: {kaitao.meng, c.masouros}@ucl.ac.uk). A. P. Petropulu is with the Department of Electrical and Computer Engineering, Rutgers University, Piscataway, NJ 08901 USA (email: athinap@rustlers.edu). L. Hanzo is with School of Electronics and Computer Science, University of Southampton, SO17 1BJ Southampton, UK (email: lh@ecs.soton.ac.uk) }
}

\maketitle


\begin{abstract}
Integrated sensing and communication (ISAC) networks are investigated with the objective of effectively balancing the sensing and communication (S\&C) performance at the network level. 
Through the simultaneous utilization of multi-point (CoMP) coordinated joint transmission and distributed multiple-input multiple-output (MIMO) radar techniques, we propose an innovative networked ISAC scheme, where multiple transceivers are employed for collaboratively enhancing the S\&C services. Then, the potent tool of stochastic geometry is exploited for characterizing the S\&C performance, which allows us to illuminate the key cooperative dependencies in the ISAC network and optimize salient network-level parameters. Remarkably, the Cramer-Rao lower bound (CRLB) expression of the localization accuracy derived unveils a significant finding: Deploying $N$ ISAC transceivers yields an enhanced average cooperative sensing performance across the entire network, in accordance with the $\ln^2N$ scaling law.
Crucially, this scaling law is less pronounced in comparison to the performance enhancement of $N^2$ achieved when the transceivers are equidistant from the target, which is primarily due to the substantial path loss from the distant base stations (BSs) and leads to reduced contributions to sensing performance gain. Moreover, we derive a tight expression of the communication rate, and present a low-complexity algorithm to determine the optimal cooperative cluster size.
Based on our expression derived for the S\&C performance, we formulate the optimization problem of maximizing the network performance in terms of two joint S\&C metrics. To this end, we jointly optimize the cooperative BS cluster sizes and the transmit power to strike a flexible tradeoff between the S\&C performance. Simulation results demonstrate that compared to the conventional time-sharing scheme or a non-cooperative scheme, the proposed cooperative ISAC scheme can effectively improve the average data rate and reduce the CRLB, hence striking an improved S\&C performance tradeoff at the network level.
\end{abstract}   

\newtheorem{thm}{\bf Lemma}
\newtheorem{remark}{\bf Remark}
\newtheorem{Pro}{\bf Proposition}
\newtheorem{theorem}{\bf Theorem}
\newtheorem{Assum}{\bf Assumption}
\newtheorem{Cor}{\bf Corollary}

\section{Introduction}

The integration of sensing and communication (ISAC) emerges as a promising paradigm for next-generation networks \cite{Liu2022SurveyFundamental, Zhang2021OverviewSignal}. It employs unified spectrum, waveforms, platforms, and networks to address the spectrum scarcity and circumvent interference caused by separate sensing and communication (S\&C) systems relying on shared wireless resources \cite{Mishra2019Toward, meng2024integrated}. Through enabling data transmission to a communication receiver while extracting target information from all the scattered echoes, ISAC can significantly enhance the spectrum-, cost-, and energy-efficiency of S\&C functionalities \cite{Cui2021Integrating}. Most existing studies on this topic in the literature primarily concentrate on the ISAC design at the link/system level, with an emphasis on waveform optimization, echo signal processing, and resource allocation for a single base station (BS) \cite{Ouyang2022Performance, Meng2023Throughput, hua2022integrated, Meng2023SensingAssisted}, while only limited attention has been dedicated to the ISAC design at the network level \cite{babu2024precoding}.

The network-level ISAC is expected to provide several pronounced benefits over conventional single-cell ISAC. In terms of sensing, the ISAC network can expand its coverage to encompass larger surveillance areas, diverse sensing angles, and richer target information. Contrary to interference-limited communications, the sensing benefits can be harnessed by forming multi-static sensing in dynamic clusters \cite{Li2023Towards, Shin2017CoordinatedBeamforming}. On the communications front, ISAC transceivers may collaboratively harness advanced coordinated multi-point (CoMP) transmission/reception techniques to connect a single user with multiple BSs for enhanced inter-cell interference management/exploitation \cite{Hosseini2016Stochastic}.
Strategically incorporating S\&C cooperation techniques holds significant promise in terms of achieving an enhanced and flexibly balanced S\&C performance within ISAC networks. Despite the above advantages, networked ISAC also brings new challenges in terms of wireless resource allocation and user/target scheduling among multiple transceivers. In addition, cooperative ISAC networks inevitably increase the signaling overhead and resource consumption due to the high demand for information exchange between transceivers. Therefore, it is necessary to strike a balance between the S\&C performance gains and control signalling costs. This motivates the exploration of innovative ISAC cooperation approaches to effectively manage/exploit the interference and to significantly improve the cooperation gain, paving the way for further ISAC network optimization. 

Effectively capturing the impact of S\&C cooperation and gaining insights into the associated network-level tradeoffs requires a precise quantitative description of the average S\&C performance across the entire cooperative ISAC network. However, this task becomes significantly more challenging in the presence of numerous random variables, including the BS locations, transmission/reception topologies, and fading values. Stochastic geometry provides a powerful mathematical tool for the analysis of multi-cell wireless networks, and has been widely used in various communication-only settings \cite{Andrews2011TractableApproach, JoHanShin2012Heterogeneous, Chen2018Stochastic}. For instance, \cite{Andrews2011TractableApproach} proposed a general framework for the analysis of the average data rate and the coverage probability in multi-cell communication networks. Moreover, \cite{wei2021performance} proposed an interference management strategy, where certain time-frequency resource blocks at certain BSs are muted based on the path loss incorporating the blockage effect, i.e., reduce interference by actively decreasing resource utilization. In addition to evaluating the performance of communication-only networks, stochastic geometry can also be employed for analyzing the sensing performance of radar and wireless sensor networks \cite{al2017stochastic, wang2023performance, Roth2019FundamentalImplications}. For example, the authors of \cite{Schloemann2016Toward} studied a sensing metric based on the signal-to-interference-plus-noise ratio (SINR) to establish a relationship between the sensing accuracy and the number of BSs transmitting signals with sufficient power to effectively participate in a localization process. In \cite{Roth2019FundamentalImplications}, by exploiting stochastic geometry methods, it was demonstrated that increasing the BS density is capable of enhancing the distance-based localization accuracy in wireless sensor networks. 

Furthermore, based on the literature of separate sensing networks and communication networks, the authors of \cite{chen2022isac} studied beam-alignment aided THz ISAC networks, where a reference signal and a synchronization signal block are jointly designed to improve the beam alignment performance, and stochastic geometry is utilized to analyze the coverage probability and network throughput. Furthermore, in \cite{Ram2022Frontiers}, the BS serves as a dual-functional transmitter that supports both radar and communication functionalities on a time-division basis. In such ISAC networks, it is verified in \cite{Ram2022Frontiers} that efficient radar detection is capable of augmenting the communication throughput in ISAC networks. In a more recent study \cite{meng2023network}, coordinated beamforming was harnessed for mitigating interference within ISAC networks, offering useful insights into the optimal allocation of spatial resources. However, the existing literature hardly touched upon mitigating the inter-cell S\&C interference by the tight cooperation of the BSs for enhancing the performance of ISAC networks.
Furthermore, it is noteworthy that the assessment of the sensing performance in the existing literature typically relies on metrics like the SINR or mutual information due to their analytical tractability \cite{chen2022isac, Ram2022Frontiers, meng2023network, Schloemann2016Toward}. But the sensing performance analysis of other metrics of estimation theory, such as the Cramér-Rao lower bound (CRLB) \cite{Eldar2006Uniformly, Lu2024Integrated}, has seldom been considered for network performance analysis, primarily due to the more complex operations involved, including matrix inversions.

Based on the above discussions, we propose a cooperative ISAC scheme that integrates CoMP-based joint transmission and multi-static sensing to strike a balance between the sensing performance and communication performance at the network level. As shown in Fig.~\ref{figure1}, multiple BSs cooperatively  transmit the same communication data to the served user, while another set of BSs collaborates with the objective of offering localization services for each target. In such cooperative ISAC networks, new degrees of freedom (DoF) can be explored for optimizing the cooperative BS cluster sizes for S\&C for achieving satisfactory S\&C performance. Too small clusters will fail to provide full cooperation gains from CoMP transmission and multi-static sensing. On the other hand, an excessive cluster size will achieve better data rate and higher sensing accuracy, albeit at the cost of additional signal processing, increased feedback, and excessive signalling overhead. Therefore, the cooperative S\&C cluster size should be dynamically adjusted to account for changes in channel conditions, user/target profiles, and service quality requirements. Moreover, optimizing the cooperative cluster sizes for S\&C should take into account the constraints imposed by the backhaul link capacity \cite{Ghimire2015Revisiting}. Intuitively, a larger cooperative cluster size brings with it higher backhaul requirements, and thus the limited backhaul capacity may become a new bottleneck, failing to enhance the cooperative S\&C performance gain, which motivates our investigation. 

In this work, we quantify the ISAC performance both in terms of the data rate and CRLB, and apply stochastic geometry techniques to conduct our performance analysis, shedding light on vital cooperative dependencies within the ISAC network. This analysis yields insights concerning both the data rate and CRLB upon increasing the cooperative S\&C cluster size. By using stochastic geometry, the S\&C performance of the entire cellular network can be described more fairly, even in the presence of an infinite number of random variables such as BS locations and fading values. Then, based on the expression derived, we aim for optimizing maximize the data-rate-CRLB region and the weighted sum of data rate and the inverse CRLB. In contrast to the most similar studies, which however dispense with cooperation \cite{olson2022coverage, salem2022rethinking}, in this work, both the cooperative cluster sizes and the power allocation of S\&C are optimized for further improving the whole network's performance and achieve a more flexible tradeoff between the S\&C performance at the network level. 
The main contributions of this paper are summarized as follows:
\begin{itemize}[leftmargin=*]
	\item Firstly, we propose a cooperative ISAC networking framework, enabling the realization of CoMP-based joint transmission and distributed radar within the constraints of limited backhaul capacity. Employing our proposed cooperative ISAC network model and stochastic geometry tools, we analytically describe the S\&C performance through tractable expressions, unveiling essential insights into the inherent dependencies within ISAC networks. Furthermore, we demonstrate that cooperative S\&C significantly enhances both the average data rate and the sensing accuracy.
	\item Secondly, upon considering random sensing directions, we reveal a squared geometry gain of $N^2$ for multi-static sensing associated with $N$ BSs. Furthermore, by relying on random BS locations, we derive tractable expressions for the expected CRLB. The scaling law of the CRLB with respect to the number of sensing BSs is established, specifically following a logarithmic square relationship, i.e., $\ln^2 N$. The primary reason for the diminished scaling law of the CRLB as compared to the geometric gain is that despite the distant BSs offering improved sensing diversity, their contributions to sensing gain is reduced due to their higher propagation loss.
	\item Thirdly, we derive the effective channel gain as well as the Laplace transform of the useful signals and of the inter-cell interference. Based on this, a tractable expression of the communication is derived for flexible cooperative cluster sizes. Moreover, the optimal cooperative communication cluster size taking into account the BS's acceptance probability is derived. 
	\item Finally, we formulate a performance boundary optimization problem for ISAC networks, whose performance is compared to an inner bound to verify that cooperative transmission and sensing in ISAC networks can effectively improve the S\&C gain and strike a more flexible tradeoff between the S\&C performance. Moreover, it is revealed by our simulations that when provided with more time/frequency resource blocks and increased backhaul capacity, the proposed cooperative scheme exhibits a higher performance improvement than the time-sharing scheme.
\end{itemize}

Notation: $B(a,b,c) = \int_0^a t^{(b-1)} (1-t)^{c-1}dt$ is the incomplete Beta function. Lowercase letters in bold font will denote deterministic vectors. For instance, $X$ and ${\bf{X}}$ denote one-dimensional (scalar) random variable and random vector (containing more than one element), respectively. Similarly, $x$ and ${\bf{x}}$ denote scalar and vector of deterministic values, respectively. ${\rm{E}}_{x}[\cdot]$ represents statistical expectation over the distribution of $x$, and $[\cdot]$ represents a variable set. 

\section{System Model}

\begin{figure}[t]
	\centering
	\includegraphics[width=8cm]{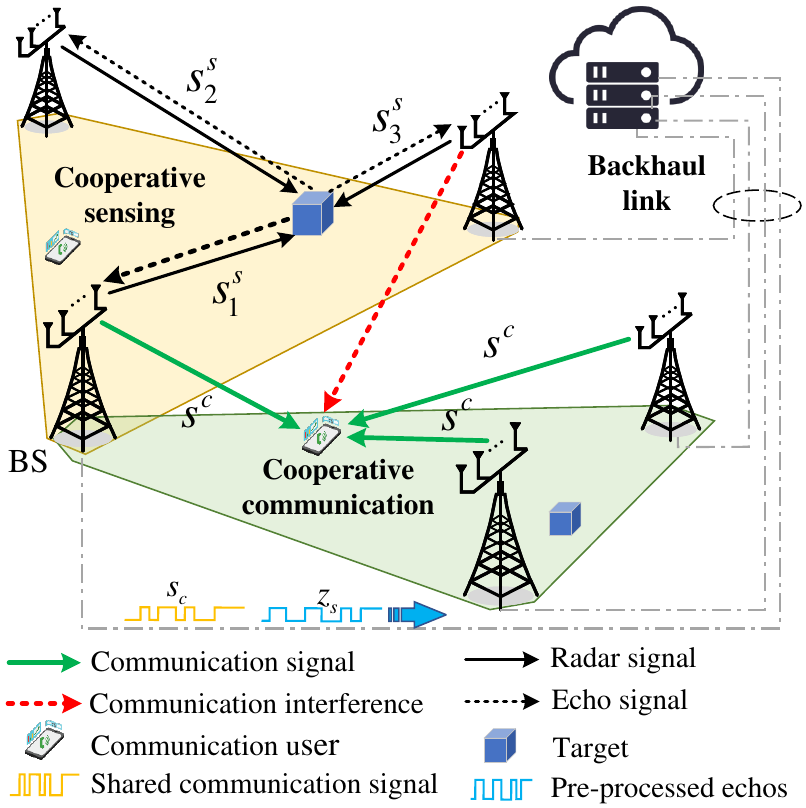}
	\vspace{-3mm}
	\caption{Illustration of cooperative S\&C networks.}
	\label{figure1}
\end{figure}
\subsection{Cooperative ISAC Network Model}
In the network considered, each BS is equipped with $M_{\mathrm{t}}$ transmit antennas and $M_{\mathrm{r}}$ received antennas. The BS' location follows a homogeneous Poisson point process (PPP) spanning the two-dimensional (2D) space, denoted by $\Phi_b$. Similarly, $\Phi_u$, and $\Phi_s$ respectively represent the point processes modeling the locations of communication users having a single antenna and of the targets in the ISAC network. It is assumed that $\Phi_b, \Phi_u$, and $\Phi_s$ are mutually independent PPPs having intensities $\lambda_b$, $\lambda_u$, and $\lambda_s$, where $\Phi_b = \{ {\bf{d}}_i \in \mathbb{R}^2, \forall i \in \mathbb{N}^+\}$, and $\lambda_u, \lambda_s \gg \lambda_b$. 
Following Slivnyak's theorem and the standard practice in stochastic geometric analysis \cite{Andrews2011TractableApproach, mukherjee2014analytical}, the typical user is assumed to be located at the origin, and its performance is analyzed for characterizing the average performance of all users \cite{Andrews2011TractableApproach}. Similarly, we assess the average sensing performance of the typical target located at the origin, where the distance from the typical target and the sensing transceivers are analyzed based on their location distribution.

Focusing on a single time-frequency resource block, each communication user is served by a cluster of $L \ge 1$ cooperative BSs that jointly transmit the same communication data, where the cooperative BS cluster is dynamically formed based on the user's location. Similarly, we adopt the target-centric clustering method, relying on $N \ge 1$ BSs collaboratively providing localization service for each target, hence forming a distributed multi-static multiple-input multiple-output (MIMO) radar system. 
Here, the BSs of the cooperative cluster can exploit the target-reflected signals over both the direct links (BS-to-target-to-originated BS links) and the cross-links (BS-to-target-to-other BSs links) for cooperative sensing, as shown in Fig.~\ref{figure1}. Each BS utilizes transmit beamforming for sending information-bearing signal $s^c$ to the corresponding served user, together with a dedicated radar signal $s^{s}_i$ for the sensed target. In line with the assumptions in \cite{Liu2020JointTransmit, Hua2023Optimal}, there is no correlation between the communication signal and the radar signal, i.e., $E\{s^s_i (s^c_i)^H\} = 0$. Upon letting ${\mathbf{s}_i=\left[s^s_i, s^c_i\right]^T}$, we have $\mathrm{E}\left[\mathbf{s}_i \mathbf{s}_i^H\right]=\mathbf{I}_2$. Then, the signal transmitted by the $i$th BS is given by
\vspace{-1.5mm}
\begin{equation}\label{TrasmitSignals}
	{\mathbf{x}}_i = \mathbf{W}_i \mathbf{s}_i = \sqrt{p^c} {\bf{w}}^c_i s^c_i +  \sqrt{p^s} {\bf{w}}^s_i s^s_i,
	\vspace{-1.5mm}
\end{equation}
where ${\bf{w}}^c_i$ and ${\bf{w}}^s_i \in {\mathbb{C}}^{M_{\mathrm{t}} \times 1}$ are normalized transmit beamforming vectors, i.e., $\|{\bf{w}}^c_i \| = 1$ and $\|{\bf{w}}^s_i\| = 1$. Furthermore, $p^s$ and $p^c$ represent the transmit power of the S\&C signals, and $\mathbf{W}_i=\left[p^c \mathbf{w}^c_i, p^s \mathbf{w}^s_i\right] \in {\mathbb{C}}^{M_{\mathrm{t}} \times 2}$ is the transmit precoding matrix of the BS at $\mathbf{d}_i$. To avoid the interference between S\&C, we adopt zero-forcing (ZF) beamforming for the sake of making the analysis tractable. Then, the beamforming vector of the serving BS $i$ is given by 
\vspace{-1.5mm}
\begin{equation}\label{TransmitBeamforming}
	{\bf{W}}_i = {{\bf{H}}_i}{\left( {\bf{H}}_i ^H {\bf{H}}_i \right)^{-1}},
	\vspace{-1.5mm}
\end{equation}
where ${\bf{H}}_i = [({\bf{h}}^H_{i,c})^T, ({\bf{a}}^H(\theta_i))^T]^T$. Here, ${\bf{h}}^H_{i,c} \in C^{M_{\mathrm{t}} \times 1}$ denotes the communication channel 
spanning from BS $i$ to the typical user, and ${{\bf{a}}^H}(\theta_i ) \in C^{M_{\mathrm{t}} \times 1}$ represents the sensing channel impinging from BS $i$ to the typical target. Upon using ZF beamforming, we have $p^s + p^c = P^{t}$, where $P^t$ is the BS transmit power.

\begin{remark}
	As illustrated in Fig. \ref{figure1}, there may be an overlap between different cooperative S\&C BS clusters. Hence, a BS within a certain communication cluster may send data to the served user, while also providing cooperative localization services within another set of BSs. Despite this overlap, each BS consistently adopts transmit beamforming following equation (\ref{TransmitBeamforming}) from its individual perspective. Consequently, the overlap does not affect the analysis of S\&C performance.
\end{remark}

\subsection{Cooperative Sensing Model}
The location of the typical target is denoted by ${{\bm{p}}}_t = [x_t, y_t]^T$. Assuming that unbiased measurements can be made, the CRLB is an exemplar benchmark for the theoretical localization accuracy in terms of the mean squared error (MSE) which can be expressed as
\vspace{-1.5mm}
\begin{equation}	
	{\rm{var}}\{\hat{{\bm{p}}}_t\} = {\rm{E}} \{| \hat{{\bm{p}}}_t - {{\bm{p}}}_t |^2\} \ge \mathrm{CRLB},
	\vspace{-1.5mm}
\end{equation}
where $\hat{{\bm{p}}}_t=\left[\hat{x}_t, \hat{y}_t\right]^T$ represents the estimated location of the typical target. In general, time synchronization in coherent distributed MIMO radar is an expensive task, but fortunately a synchronization error  corresponding to the propagation delay of the cell is adequate for non-coherent processing \cite{sadeghi2021target}. Hence, non-coherent MIMO radar is more practical and thus it is assumed in this work. The typical target is collaboratively sensed by $N$ BSs. Let us assume that the transmitted radar signals $\{s_i^s\}_{i=1}^N$ of the BSs in the cooperative sensing cluster are approximately orthogonal for any time delay of interest \cite{li2008mimo}.\footnote{Due to the substantial amount of original echo data, the echo signals received at each BS can be pre-processed (e.g., coherent operations) in a distributed manner, and then send the processing results to the central unit through backhaul links, which can significantly reduce the backhaul overhead.} For non-coherent MIMO radar, the base-band equivalent of the impinging signal at receiver $j$ is represented as 
\begin{equation}\label{SensingChannel}
	\begin{aligned}
		{{\bf{y}}_{j}}(t) =& \sum\nolimits_{i = 1}^N \sigma \underbrace {{{\left\| {{{\bf{d}}_j}} \right\|}^{ - \frac{\beta}{2} }}{\bf{b}}\left( {{\theta _j}} \right){{\left\| {{{\bf{d}}_i}} \right\|}^{ - \frac{\beta}{2} }}{{\bf{a}}^H}\left( {{\theta _i}} \right)}_{{\text{target channel}}}\sqrt{p^s}{{\bm{w}}^s_i}\\
		 &\times {s^s_i}\left( {t - {\tau _{i,j}}} \right) + {\bf{n}}_l(t),
	\end{aligned}
\end{equation}
where $\beta \ge 2$ is the pathloss exponent between the serving BS and the typical target, $\sigma$ denotes the radar cross section (RCS), $\tau _{i,j}$ is the propagation delay of the link spanning from BS $i$ to the typical target and then to BS $j$. Finally, the term ${\bf{n}}_l(t)$ is the additive complex Gaussian noise having zero mean and covariance matrix ${\bm{\Sigma}} = \sigma_s^2 {\bf{I}}_{M_{\rm{r}}}$. In (\ref{SensingChannel}), we have ${{\bf{a}}^H}(\theta_n ) = [1, \cdots, e^{ {j \pi(M_{\mathrm{t}}-1)  \cos(\theta_n) }}]^T$, and ${\bf{b}}(\theta_j ) = [1, \cdots, e^{ {j \pi(M_{\mathrm{r}}-1)  \cos(\theta_j) }}]$, where $\theta_i$ is the angle of bearing for the $i$-th BS to the target with respect to the horizontal axis. 

By evaluating the range of each monostatic link and bi-static link, the target location can be estimated by maximum likelihood estimation (MLE) \cite{li1993maximum}. Then, the Fisher information matrix (FIM) of estimating the parameter vector ${{\bm{p}}}_t$ for the non-coherent MIMO radar considered is equal to \cite{sadeghi2021target}
\begin{equation}\label{FIMexpression}
{\bf{F}}_N = |\zeta |^2 \sum\nolimits_{i = 1}^N \sum\nolimits_{j = 1}^N {\left\| {{{\bf{d}}_i}} \right\|}^{ - \beta }{\left\| {{{\bf{d}}_j}} \right\|}^{ - \beta } {  {\left[ {\begin{array}{*{20}{c}}{a_{ij}^2}&{{a_{ij}}{b_{ij}}}\\
		{{a_{ij}}{b_{ij}}}&{b_{ij}^2}
\end{array}} \right]} } ,
\vspace{-1.5mm}
\end{equation}
where 
\vspace{-1.5mm}
\begin{equation}\label{a_ijDefinition}
	{a_{ij}} = \cos {\theta _i} + \cos {\theta _j},
\end{equation}
\begin{equation}\label{b_ijDefinition}
	{b_{ij}} = \sin {\theta _i} + \sin {\theta _j}.
\end{equation}
In (\ref{FIMexpression}), we have \cite{sadeghi2021target}
\vspace{-1.5mm}
\begin{equation}
	|\zeta |^2 = \frac{p^s G_t G_r B^2 \sigma}{8 \pi f_c^2 \sigma_s^2}, 
	\vspace{-1.5mm}
\end{equation}
where $G_t$ and $G_r$ denote the transmit beamforming gain and receive beamforming gain, $f_c$ is the carrier frequency, and $B^2$ represents the squared effective bandwidth. Furthermore, $|\zeta |$ represents the common signal processing gain term.
We assume that the target is stationary during the short signal processing interval, and thus the RCS distribution obeys a Swerling-I model \cite{skolnik1980introduction}, i.e., $p(\sigma) = \frac{1}{\sigma_{av}} e^{-\frac{\sigma}{\sigma_{av}}}$,
where $\sigma_{av}$ is the mean value of the target RCS. Since the elements of the FIM in (\ref{FIMexpression}) are random due to the fluctuation of the target RCS, this makes the CRLB analysis challenging. Thus, we aim for analyzing the FIM for the average RCS for ease of analysis. Then, the FIM of the average RCS is defined as $\bar {\bf{F}}_N = {\rm{E}}_{\sigma} \left\{ {\bf{F}}_N \right\}$. Given the random location of ISAC transceivers, the expected CRLB for any unbiased estimator of the target position is given by
\vspace{-1.5mm}
\begin{equation}	
	 \mathrm{CRLB}= {\rm{E}}_{\Phi_b} \left[\operatorname{tr}\left(\bar {\mathbf{F}}_N^{-1}\right)\right].
	 \vspace{-1.5mm}
\end{equation}

\subsection{Cooperative Communication Model}

There are two types of joint CoMP transmission: coherent and non-coherent transmission. In non-coherent transmission, each BS of the cooperation set only uses its own channel information. This regime is supported by 3GPP LTE Release 11. By contrast, coherent transmission assumes the knowledge of the channel state information (CSI) for all the serving links to enable joint transmit precoding associated with perfect phase alignment. Naturally, this inevitably introduces excessive transmission overhead as well as delay. Hence, it is not explicitly supported by 3GPP. Therefore, in this work, we consider the more practical non-coherent joint transmission. We adopt user-centric clustering methods, where the BS closest to the typical user sends collaboration service requests to the other $L$ BSs closest to the typical user. The set of BSs receiving the service request from the closest BS is denoted by $\Phi_c$. Each BS in this cluster decides whether to accept the request based on its traffic load, which will be analyzed in Lemma \ref{AcceptationProbabilityC}. The set of BSs that do accept the collaboration request is denoted by $\Phi_a$.  

The index of the closest BS to the typical user is 1. The large-scale pathloss between the user and the closest BS obeys $\left\|\mathbf{d}_1\right\|^{-\alpha}$, where $\mathbf{d}_1$ and $\left\|\mathbf{d}_1\right\|$ respectively represents the location of the closest BS and its distance from the typical user, while $\alpha \ge 2$ is the pathloss exponent. Then, the received signal of the typical user is given by 
\vspace{-1.5mm}
\begin{equation}
	\begin{aligned}
		y_{c}=& \underbrace{\left\|\mathbf{d}_1\right\|^{-\frac{\alpha}{2}} \mathbf{h}_{1}^H \mathbf{W}_1 \mathbf{s}_1}_{\text{intended signal}} + \underbrace{\sum\nolimits_{i \in \Phi_a }\left\|\mathbf{d}_i\right\|^{-\frac{\alpha}{2}} \mathbf{h}_{i}^H \mathbf{W}_i \mathbf{s}_1}_{\text{collaborative intended signal}}  \\
		&+\underbrace{\sum\nolimits_{{j \in \{\Phi_b \backslash \Phi_a \backslash \{1\}\}}}\left\|\mathbf{d}_j\right\|^{-\frac{\alpha}{2}} \mathbf{h}_{j}^H \mathbf{W}_j \mathbf{s}_j}_{\text{inter-cluster interference}} + \underbrace{n_{c}}_{\text{noise}},
		\vspace{-1.5mm}
	\end{aligned}
\end{equation}
where $\mathbf{h}^H_{i} \sim \mathcal{C N}\left(0, \mathbf{I}_{M_{\mathrm{t}}}\right)$ is the channel vector of the link between the BS at $\mathbf{d}_i$ to the typical user, and $\Phi_a$ is the cooperative BS set. We focus on evaluating the performance of an interference-limited network within dense cellular scenarios. The impact of noise is disregarded in this analysis. The evaluation is based on the signal-to-interference ratio (SIR) \cite{Park2016OptimalFeedback}.
The SIR of the received signal at the typical user can be expressed as
\vspace{-1.5mm}
\begin{equation}\label{SIRexpression}
	{\rm{SI}}{{\rm{R}}_c} = \frac{ {{g_1}{{\left\| {{{\bf{d}}_1}} \right\|}^{ - \alpha }}} +{\sum\limits_{i \in \Phi_a} {{g_i}{{\left\| {{{\bf{d}}_i}} \right\|}^{ - \alpha }}} }}{{\sum\limits_{j \in \{\Phi_b \backslash \Phi_a \backslash \{1\}\}}  {{g_j}} {{\left\| {{{\bf{d}}_j}} \right\|}^{ - \alpha }}}},
	\vspace{-1.5mm}
\end{equation}
where $g_{1}= p^c\left|\mathbf{h}_{1}^H \mathbf{w}_1^c\right|^2$ and $g_{i}=p^c\left|\mathbf{h}_{i}^H \mathbf{w}_i^c\right|^2$ denotes the effective desired signals' channel gain, ${\sum\limits_{i \in \Phi_a} {{g_i}{{\left\| {{{\bf{d}}_i}} \right\|}^{ - \alpha }}} }$ is the cooperative desired signal, ${\sum\limits_{j \in \{\Phi_b \backslash \Phi_a \backslash \{1\}\}}  {{g_j}} {{\left\| {{{\bf{d}}_j}} \right\|}^{ - \alpha }}}$ represents the inter-cluster interference. Finally, the interference channel's gain is $g_{j} = p^c\left|\mathbf{h}_{j}^H \mathbf{w}_j^c\right|^2 + p^s\left|\mathbf{h}_{j}^H \mathbf{w}_j^s\right|^2$. The average data rate of users is given by 
\vspace{-1.5mm}
\begin{equation}
	R_c=\mathrm{E}_{\Phi_b,g_i}[\log (1+\mathrm{SIR}_c)].
	\vspace{-1.5mm}
\end{equation}

\subsection{Limited Backhaul Capacity Model}
In the system considered, each BS is connected to the central unit through backhaul links to share both the data information and statistical average CSI for CoMP transmission and to collect the echo signals for cooperative sensing, as shown in Fig.~\ref{figure1}. 
It is noteworthy that in a cooperative S\&C system, the sizes of cooperative clusters are practically restricted by the constrained capacity of the backhaul link. This limitation arises because the data volume of both information sharing and the echo signal collection escalation as the cluster size. These exchanges occur through the same link connecting each BS to the central unit. In this scenario, the realistic backhaul constraints introduce an additional tradeoff in striking a balance between the S\&C performance. In this first study of the system, we exclude consideration of the backhaul traffic related to CSI sharing \cite{Zhang2013Downlink}, focusing solely on evaluating the influence of data sharing and echo signal transmission on the constrained capacity of the backhaul link, yielding a best-case performance given by
\vspace{-1.5mm}
\begin{equation}
	R_c + e \times N \le C_{\text{backhaul}},
	\vspace{-1.5mm}
\end{equation}
where $R_c$ is the average data rate, $e$ represents the data rate required for the sensing cooperation, i.e. for sending pre-processed results (i.e., auto-correlation of the signal $\{s^s_i\}_{i=1}^N$ transmitted by each BS), and $C_{\text{backhaul}}$ denotes the backhaul capacity limitation.

To strike a flexible tradeoff between the S\&C performance at the network level, we aim for optimizing the cooperative cluster sizes of S\&C under the above backhaul capacity constraints. Our approach focuses on an analytical method for determining the optimal cooperative cluster sizes. This makes it more practically appealing compared to existing clustering schemes employing iterative algorithms, which are more complex and computationally heavy. The optimal cluster parameter can be readily calculated offline from the prior knowledge of the BS load, the ratio of BSs to user densities, and path loss exponents. Since the parameters $L$ and $N$ do not vary significantly, they can still be updated online, reducing both the system overhead and computational complexity compared to existing iterative algorithms.

\section{Sensing Performance Analysis}
\label{SensingSection}

In this section, we first analyze the sensing geometry-based gain and localization accuracy gain attained by the cooperative BSs, when there are sufficient resource blocks. In this case, each BS always accepts the cooperative localization requests. Then, in Section \ref{AccperationProbability}, we analyze the optimal cluster size of cooperative BSs subject to limited resource block availability.

\subsection{Geometry Gain of Cooperative Sensing}

In practice, it is non-trivial to derive a tractable expression for the expected CRLB due to the complex operations involved, including matrix inversions. 
To this end, following the process adopted in \cite{sadeghi2021target}, we first study the sensing geometry-based gain where all transceivers are assumed to be placed at the same distance from the target, i.e. $d_i = d_j$, $\forall i,j, \in {\cal{N}}$, and ${\cal{N}} = \{1,\cdots,N\}$. Then, by ignoring the measurement gain of each link, i.e., $|\zeta |$ and $d_i$, we analyze the cooperative localization performance improvement purely gleaned from the direction at diversity of transceivers. To this end, we define a new matrix based on the FIM by removing $|\zeta |$ and $d_i$ in (\ref{FIMexpression}), yielding
\vspace{-2mm}
\begin{equation}
	\tilde {\bf{F}}_N = \sum\nolimits_{i = 1}^N {\sum\nolimits_{j = 1}^N {\left[ {\begin{array}{*{20}{c}}
					{a_{ij}^2}&{{a_{ij}}{b_{ij}}}\\
					{{a_{ij}}{b_{ij}}}&{b_{ij}^2}
			\end{array}} \right]} } .
		\vspace{-1.5mm}
\end{equation}
Following the definition in \cite{guvenc2009survey}, let ${\rm{tr}}({ \tilde {\bf{F}}^{-1}_N})$ be termed the geometric dilution of precision (GDoP) in non-coherent distributed MIMO radar, which can be formulated as follows:
\vspace{-1.5mm}
\begin{equation}\label{GDoPExpression}
	\begin{aligned}
		&{\rm{tr}}({\tilde {\bf{F}}}_N^{-1}) = \\
		&\frac{\sum\nolimits_{i = 1}^N {\sum\nolimits_{j = 1}^N {a_{ij}^2} }  + \sum\nolimits_{i = 1}^N {\sum\nolimits_{j = 1}^N {b_{ij}^2} } }{\left(\! {\sum\nolimits_{i = 1}^N {\sum\nolimits_{j = 1}^N {a_{ij}^2} } } \!\right)\!\left(\! {\sum\nolimits_{i = 1}^N {\sum\nolimits_{j = 1}^N {b_{ij}^2} } } \right) \! - \! {\left( {\sum\nolimits_{i = 1}^N {\sum\nolimits_{j = 1}^N {{a_{ij}}{b_{ij}}} } } \!\right)^2}}.
		\vspace{-1.5mm}
	\end{aligned}
\end{equation}
Providing a direct characterization of the geometry-based gain resulting from cooperative localization remains challenging, primarily due to the unpredictable nature of sensing directions. Hence, we adopt some tight approximations to derive a tractable expression, aiming for offering an intuitive illustration of the geometry-based gain.

\begin{Pro}\label{GDoPDerivation}
	The expected GDoP can be approximated as
	\vspace{-1.5mm}
	\begin{equation}\label{GDoPExpression1}
		\begin{aligned}
			{\rm{E}}_{\theta}\left[ {{\rm{tr}}\left( {\tilde {\bf{F}}}_N^{-1} \right)} \right] \approx \frac{{2N + 2}}{{{N^3} - {N^2}}}.
			\vspace{-1.5mm}
		\end{aligned}
	\end{equation}
\end{Pro}
\begin{proof}
	Please refer to Appendix A.
\end{proof}

Building upon the conclusion in Proposition \ref{GDoPDerivation}, the scaling law associated with an infinite number of ISAC BSs can be derived as follows.
\begin{Cor}{\label{ScalingLawWithoutDis}}
	For an infinite number of BSs involved in cooperative sensing, the expected GDoP can be further reformulated from (\ref{GDoPExpression1}) as
	\vspace{-1.5mm}
	\begin{equation}\label{GeometryGain}
		\mathop {\lim }\limits_{N \to \infty } {\rm{E}}_{\theta}\left[ {{\rm{tr}}\left( {\tilde {\bf{F}}}_N^{-1} \right)} \right] = \mathop {\lim }\limits_{N \to \infty } \frac{{2N + 2}}{{{N^3} - {N^2}}} = \frac{2}{{{N^2}}}.
		\vspace{-1.5mm}
	\end{equation}
\end{Cor}
Corollary \ref{ScalingLawWithoutDis} states that the expected GDoP is inversely proportional to the square of the number of BSs involved. The analysis above provides insights into the geometric-based gain for cooperative sensing by only considering random sensing directions. In the following section, we will further detail the expected CRLB derivation and compare the localization gain and geometry-based gain.

\subsection{Performance Gain of Cooperative Sensing}
In this subsection, we derive the closed-form CRLB expression under the assumption of random locations of both the BSs and targets. Based on this, the scaling law of localization accuracy may be obtained. First, the CRLB expression can be equivalently transformed into
\vspace{-1.5mm}
\begin{equation}
	\begin{aligned}
		&{\rm{CRLB}} = {{\rm{E}}_{\Phi_b}}\bigg[|\zeta |^{-2} \times \\
		& \frac{{2\sum\nolimits_{i = 1}^N \!{\sum\nolimits_{j = 1}^N d_i^{-\beta}d_j^{ - \beta}\left(1+\cos \left( {{\theta _i} - {\theta _j}} \right)\right) } } } {{\!\sum\nolimits_{l = 1}^N {\! \sum\nolimits_{k = 1}^N {\!\! \sum\nolimits_{i \ge k}^N {\! \sum\nolimits_{j > \!{\lceil (k - i)N + l \rceil \!}^+}^N \!{(d_i d_j d_l d_k)^{-\beta}} } } } {{\! \left( {{a_{kl}}{b_{ij}} \!-\! {a_{ij}}{b_{kl}}} \right)}^2}}}\! \bigg] \!,
		\vspace{-1.5mm}
	\end{aligned}
\end{equation}
where $d_i = {\left\| {{{\bf{d}}_i}} \right\|}$ and ${\lceil x \rceil}^+ = \max(x,1)$.
To obtain a more tractable CRLB expression, we resort to a simple yet tight approximation. Then the following conclusion is proved.
\begin{Pro}\label{SimplifiedWithDis1}
The expected CRLB can be approximated as 
\vspace{-1.5mm}
\begin{equation}\label{SimplifiedExpressionCRLB}
	{\rm{CRLB}} =	\frac{2}{|\zeta |^{2}{\sum\nolimits_{l = 1}^N {\sum\nolimits_{k = 1}^N {E{{\left[ {{d_k}} \right]}^{ - \beta}}E{{\left[ {{d_l}} \right]}^{ - \beta}}} } }}.
	\vspace{-1.5mm}
\end{equation}
\end{Pro}
\begin{proof}
	Please refer to Appendix B.
\end{proof}

Interestingly, we found that the expected CRLB in Proposition \ref{SimplifiedWithDis1} is only determined by the expected distance from the BS to the typical target. It maybe readily verified that (\ref{SimplifiedExpressionCRLB}) achieves a good approximation by Monte Carlo simulations, as shown in Section \ref{simulations}.
Furthermore, the expected distance from the $n$th closest BS to the typical target can be expressed as
\vspace{-1.5mm}
\begin{equation}\label{ExpectedDistance}
	E \left[ {{d_n}} \right] =  { {\frac{{\Gamma \left( n + \frac{1}{2} \right)}}{{\sqrt{\lambda_b \pi} \Gamma (n)}}} } \approx \sqrt{\frac{{n}}{\lambda_b \pi}}.
	\vspace{-1.5mm}
\end{equation}
By substituting (\ref{ExpectedDistance}) into (\ref{SimplifiedExpressionCRLB}), the CRLB expression can be further approximated as
\vspace{-1.5mm}
\begin{equation}\label{CRLB_expression}
	{\rm{CRLB}} \approx \frac{2}{|\zeta |^{2}{{\lambda_b ^\beta}{\pi ^\beta}\sum\nolimits_{l = 1}^N {\sum\nolimits_{k = 1}^N { k^{-\frac{\beta}{2}} l^{-\frac{\beta}{2}}} } }}.
	\vspace{-1.5mm}
\end{equation}
For $\beta = 2$, we further derive the scaling law of the localization accuracy as follows.

\begin{theorem}\label{SimplifiedWithDis3}
	For an infinite cooperative cluster size $N$, the expected CRLB is given by
	\vspace{-1.5mm}
	\begin{equation}
		\mathop {\lim }\limits_{N \to \infty }{\rm{CRLB}} \times {{\ln }^2}N= \frac{1}{|\zeta |^{2}{{\lambda_b ^2}{\pi ^2}}}.
		\vspace{-1.5mm}
	\end{equation}
\end{theorem}
\begin{proof}
	Please refer to Appendix C.
\end{proof}

\begin{remark}
	The CRLB scaling law of the random BS location formulated in Theorem \ref{SimplifiedWithDis3} is useful for cooperative sensing design. In contrast to the case that arises in Proposition \ref{GDoPDerivation}, as the number of participating BSs increases, the performance gain diminishes. The main reason is that even though the distant BSs do provide improved sensing diversity, these measurements only modestly contribute to the localization accuracy gain due to the excessive propagation loss.
	According to Theorem \ref{SimplifiedWithDis3}, it can be found that the root mean square error follows $\sqrt{{\rm{CRLB}}} \sim \frac{1}{{{{\ln }}N}}$. To the best of our knowledge, this is the first work to derive the scaling law of localization accuracy CRLB for random sensing transceiver locations.
\end{remark}

\subsection{Cooperation Acceptance Probability for Limited Resource Blocks}
\label{AccperationProbability}
Using the target-centric clustering model of \cite{Li2015UserCentric}, the center unit sends localization service requests to the $N$ closest BSs in the vicinity of the typical target. In practice, enlarging the cluster size may lead to some BSs declining requests due to their limited availability of time/frequency resource blocks. Let $\psi$ be an integer representing the maximum load, i.e., the maximum number of targets that can be simultaneously served by the BS.
Then, if a BS receives $N$ requests, it will randomly choose $\psi$ targets to provide services to. In this case, the cooperation acceptance probability of the BS which receives service requests is given by Lemma \ref{AcceptationProbability}:
\begin{thm}\label{AcceptationProbability}
	When each target requests $N$ BSs to provide localization services, the acceptance probability of the BS can be expressed by:
	\vspace{-1.5mm}
	\begin{equation}\label{AcceptationSensing}
		\kappa_s = \frac{\Gamma {\left( \psi ,\mu_s \bar N \right)}}{{(\psi  - 1)!}} + \sum\limits_{n = \psi  + 1}^\infty  {\frac{{ {\psi} {{\left( {\mu_s \bar N} \right)}^n}}}{{n \times n!}}} {e^{ - \mu_s \bar N}},
		\vspace{-1.5mm}
	\end{equation}
	where $\mu_s = {\frac{{{\lambda _s}}}{{{\lambda _b}}}}$ and $\bar N = {\frac{\Gamma(N+\frac{1}{2})^2}{\Gamma(N)^2}}$.
\end{thm}
\begin{proof}
	Please refer to Appendix D.
\end{proof}

According to Lemma \ref{AcceptationProbability}, the acceptance probability is monotonically increasing with the number of resource blocks $\psi$, and it is monotonically decreasing with both the target-BS density ratio $\mu_s$ and the cluster size $N$. We assume that the number of requests received by different BSs is an independent random variable.  
Then, with $\beta = 2$, the expected CRLB under the consideration of the acceptance probability can be expressed as 
\vspace{-1.5mm}
\begin{equation}
	{\rm{CRLB}}_a =  \frac{1}{\kappa_s ^2 |\zeta |^{2}{{ \lambda_b ^2}{\pi ^2}{{\ln }^2}N}}.
	\vspace{-1.5mm}
\end{equation}

It can be inferred that $\kappa$ is decreasing with $L$, and we have $\kappa = 1$ when ${\frac{\lambda_u \bar L}{\lambda_b}} \ll \psi$. Hence, the ${\rm{CRLB}}_a$ will first decrease and then increase upon increasing $N$.

\begin{remark}\label{RatioOptimal}
	In general, the optimal cluster size generally satisfies that $\kappa_s \approx 1$ and ${\frac{\lambda_u \bar L}{\lambda_b}} \le \psi$. The primary rationale behind this lies in the fact that the average number of services provided by each BS throughout the entire network must not surpass $\psi$. When the average number of collaboration requests ${\frac{\lambda_u \bar L}{\lambda_b}}$ exceeds $\psi$, $\kappa_s$ becomes less than one, indicating that the BS must deny some requests. As a result, the average distance of participating BSs from the typical target in the cooperative cluster increases, thereby eroding the collaboration gain.
\end{remark}

\section{Communication Performance Analysis}
\label{CommunicationSection}
In this section, we derive a tractable expression of communication performance, and present an approximated expression to acquire the optimal cluster size for cooperative transmission.

\subsection{Expression of Communication Rate}
\label{CommunicationPerformance}
To support CoMP-based joint transmission, the closest BS sends service requests to the other $L-1$ BSs closest to the typical user. Similarly, if a BS receives more than $\psi$ requests, we assume it will randomly choose $\psi-1$ users for providing services to besides the typical user.\footnote{The closest BS always provides service to the typical user since it sends the cooperation request.} In this case, the analysis of the acceptance probability for service requests received by the BSs is formulated as follows:
\begin{thm}\label{AcceptationProbabilityC}
	When each user requests $L$ BSs to provide communication services for him/her, the acceptance probability of the BS can be formulated as:
	\vspace{-1.5mm}
	\begin{equation}
		\kappa_c = \frac{{\Gamma \left( {\psi ,\mu_c \bar L} \right)}}{{(\psi  - 1)!}} + \sum\limits_{n = \psi}^\infty  {\frac{{\left( {\psi  - 1} \right){{\left( {\mu_c \bar L} \right)}^n}}}{{(n-1) \times n!}}} {e^{ - \mu_c \bar L}},
		\vspace{-1.5mm}
	\end{equation}
	where $\mu_c = \frac{\lambda_u}{\lambda_b}$ and $\bar L = {\frac{\Gamma(L+\frac{1}{2})^2}{\Gamma(L)^2}}$.
\end{thm}
\begin{proof}
	The proof can be completed in a similar way as presented in Appendix D, but the detailed proof is omitted due to page limitation.
\end{proof}

According to \cite{hamdi2010useful}, for the uncorrelated variables $X$ and $Y$, we have
\vspace{-1.5mm}
\begin{equation}\label{CommunicationBasicEquation}
	{\rm{E}}\left[ {\log \left( {1 + \frac{X}{Y}} \right)} \right] \! = \! \int_0^\infty  {\frac{1}{z}} \left( {1 - {\rm{E}}\left[{e^{ - z \left[ X\right] }}\right]} \right){\rm{E}}\left[{e^{ - z\left[ Y \right]}}\right]{\rm{d}}z.
	\vspace{-1.5mm}
\end{equation}
In (\ref{CommunicationBasicEquation}), ${\rm{E}}\left[{e^{ - z \left[ X\right] }}\right]$ and ${\rm{E}}\left[{e^{ - z \left[ Y\right] }}\right]$ are the Laplace transforms of $X$ and $Y$.
Then, under a given distance $r$ from the typical user to the closest BS, the conditional expectation of data rate can be expressed as follows:
\vspace{-1.5mm}
\begin{equation}
	\begin{aligned}
		&{\rm{E}}\left[ {\log \left( {1 + \mathrm{SIR}_c} \right)} \big| r \right] \\
		=& {\rm{E}} \! \left[ {\log \left( \! {1 + \frac{{g_{1}} + \sum\nolimits_{{{i}} \in {\Phi_a}} {g_{i}} \left\| {\bf{d}}_i \right\|^{-\alpha} r^{\alpha}}{ \sum\nolimits_{{{j}} \in \{\Phi_b \backslash \Phi_a \backslash \{1\}\}} g_j \left\| {\bf{d}}_j \right\|^{-\alpha} r^\alpha }} \right)} \right] \\
		=& \int_0^\infty  {\frac{{1 - {\rm{E}}\left[ {{e^{ - zg_{1}}}} \right] {\rm{E}}\left[ {{e^{ - z U}}} \right]}}{z}} {\rm{E}}\left[ {{e^{ - z I_{1}}}} \right] {\rm{E}}\left[ {{e^{ - z I_{2}}}} \right]{\rm{d}}z,
		\vspace{-1.5mm}
	\end{aligned}
\end{equation}
where $U = \sum\nolimits_{{{i}} \in {\Phi_a}} {g_{i}} r^{\alpha}$, $ I_{1} = \sum\nolimits_{{{i}} \in \{\Phi_c \backslash \Phi_a\}}  {{g_{i}}} {{\left\| {{{\bf{d}}_{{i}}}} \right\|}^{ - \alpha }}{r^\alpha }$, and $ I_{2} = \sum\nolimits_{{{i}} = L+1}^\infty  {{g_{i}}} {{\left\| {{{\bf{d}}_{{i}}}} \right\|}^{ - \alpha }}{r^\alpha }$. Here, $I_1$ represents the interference emanating from the BSs declining the cooperation requests, and $I_2$ represents the interference arising from the BS located beyond the cooperative request cluster. Furthermore, $g_{1}$ and $g_{i}$ are the effective desired signal's channel gain, where $g_{1}, g_{i} \sim \Gamma \left( M_{\mathrm{t}} - 1, p^c\right)$ \cite{Hosseini2016Stochastic}. According to the definition below (\ref{SIRexpression}), we can derive the distribution of $g_{j}$ based on the moment matching technique \cite{Hosseini2016Stochastic}. First, due to ${\rm{E}}[p^c\left|\mathbf{h}_{j}^H \mathbf{w}_j^c\right|^2] = p^c$ and ${\rm{E}}[p^s\left|\mathbf{h}_{j}^H \mathbf{w}_j^s\right|^2] = p^s$, we have ${\rm{E}}[g_{j}] = p^s + p^c = P^t$. Then, ${\rm{E}}[g_{j}^2] = {\rm{E}}\left[\left|\mathbf{h}_{j}^H \mathbf{w}_j^s\right|^4\right] + {\rm{E}}\left[\left|\mathbf{h}_{j}^H \mathbf{w}_j^c\right|^4\right]+2{\rm{E}}\left[\left|\mathbf{h}_{j}^H \mathbf{w}_j^s\right|^2 \left|\mathbf{h}_{j}^H \mathbf{w}_j^c\right|^2\right] = (p^s + p^c)^2 = (P^t)^2$ since ${\rm{E}}\left[\left|\mathbf{h}_{j}^H \mathbf{w}_j^s\right|^2 \left|\mathbf{h}_{j}^H \mathbf{w}_j^c\right|^2\right] = {\rm{E}}\left[\left|\mathbf{h}_{j}^H \mathbf{w}_j^s\right|^2\right] {\rm{E}}\left[\left|\mathbf{h}_{j}^H \mathbf{w}_j^c\right|^2\right]$. 
Therefore, the interference channel gain $g_{j}$ can be approximated by a gamma distributed random variable having a shape parameter of $1$ and scale parameter of $P^t$. Therefore, it follows that $g_{j} \sim \Gamma \left(1 , P^t\right)$.

Based on the above discussions, the useful signal power can be expressed by 
\vspace{-1.5mm}
\begin{equation}\label{UsefulSignalPower}
	\begin{aligned}
		{\rm{E}}\left[ {{e^{ - zg_{1}}}} \right] &\simeq \int_0^\infty  {\frac{{{e^{ - zx}}{x^{{M_{\mathrm{t}} - 2 }}}{e^{-\frac{x}{p^c}}}}}{{(p^c)^{M_{\mathrm{t}} - 1}\Gamma \left( {M_{\mathrm{t}} - 1} \right)}}} {\rm{d}}x \\
		&= {\left( {1 + p^cz} \right)^{1 - {M_{\mathrm{t}} }}},
		\vspace{-1.5mm}
	\end{aligned}
\end{equation}
where ${M_{\mathrm{t}} - 1}$ refers to the diversity gain provided for the serving user.
Then, we derive tight bounds on the Laplace transform of the cooperative transmission power and on the communication interference as follows:

\begin{figure*}
	\begin{align}\label{TightCommunicationExpression}
		R_c =& \int_0^\infty  \int_0^1 \frac{2\left( {L - 1} \right)\eta_L {\left( {1 - {\eta_L ^2}} \right)^{L - 2}}}{z}\bigg( \frac{1}{{\left( {1 - \kappa_c } \right){{\rm{H}}_1}\left( {z P^t,1,\alpha,\eta_L } \right) + {{\rm{H}}_2}\left( {z P^t,\alpha,\eta_L } \right) + 1}} -   \nonumber \\
		&\frac{{{{\left( {1 + p^c z} \right)}^{ 1- {M }_{\mathrm{t}} }}}}{{\kappa_c {{\rm{H}}_1}\left( {z p^c,M_{\mathrm{t}} - 1,\alpha,\eta_L } \right) + \left( {1 - \kappa_c } \right){{\rm{H}}_1}\left( {z P^t,1,\alpha,\eta_L } \right) + {{\rm{H}}_2}\left( {z P^t,\alpha,\eta_L } \right) + 1}} \bigg){\rm{d}}\eta_L {\rm{d}}z.
		\vspace{-3mm}
	\end{align}
	\vspace{-3mm}
	\hrulefill
\end{figure*}

\begin{thm}\label{LaplaceTransform}
	With the closest BS at a distance $r$, the Laplace transforms of $U$, $I_1$, and $I_2$ are given by
	\vspace{-1.5mm}
	\begin{equation}
		{\rm{E}}\!\left[ {{e^{ - z U}}} \right] \!=\! \exp \!\bigg( \! - \pi \kappa_c \lambda_b {r^2}{\rm{H}}_1\left( { zp^c,M_{\mathrm{t}}-1,\alpha ,\eta_{L} } \right) \!\bigg),
	\end{equation}
	\begin{equation}
		{\rm{E}}\left[ {{e^{ - zI_1}}} \right] =  \! \exp \bigg( \! - \pi \left(1-\kappa_c \right) \lambda_b {r^2}{\rm{H}}_1\left( {zP^t,1,\alpha ,\eta_{L} }  \!\right) \! \bigg),
	\end{equation}
	\begin{equation}
		{\rm{E}}\left[ {{e^{ - zI_2}}} \right] = \exp \bigg(  - \pi \lambda_b {r^2}{\rm{H}}_2\left( {z P^t,\alpha ,\eta_{L} } \right) \bigg),
	\end{equation}
where ${\rm{H}}_1\left( {x,K,\alpha ,\eta_L } \right)  = \frac{1}{{{\eta ^2}}}\left( {1 - \frac{1}{{{{\left( {1 + x{\eta ^\alpha }} \right)}^K}}}} \right) + \frac{1}{{{{\left( {1 + x} \right)}^K}}} - 1 + K{x^{\frac{2}{\alpha }}}\!\left(\! {B\left( {\frac{x}{{x + 1}},1 \! - \! \frac{2}{\alpha },K + \frac{2}{\alpha }} \right) \! - \! B\left( {\frac{{x{\eta ^\alpha }}}{{x{\eta ^\alpha } + 1}},1 - \frac{2}{\alpha }, K + \frac{2}{\alpha }} \right)} \right)$, ${\rm{H}}_2\left( {x,\alpha ,\eta_L } \right) = {x^{\frac{2}{\alpha }}}B\left( {\frac{x}{{x + {\eta_L ^{ - \alpha }}}},1 - \frac{2}{\alpha },1 + \frac{2}{\alpha }} \right) + \frac{1}{{{\eta_L ^2}}}\left( {{{{{\left( {1 + x{\eta_L ^\alpha }} \right)}^{-1}}}} - 1} \right)$, $\eta_{L} = \frac{r}{r_L}$, and $r_L = \|{\bf{d}}_L \|$.
\end{thm}
\begin{proof}
	Please refer to Appendix E.
\end{proof}

Based on the Laplace transforms of $U$, $I_1$, and $I_2$ obtained, the expected data rate is formulated in Theorem \ref{CommunicationTightExpression}.

\begin{theorem}\label{CommunicationTightExpression}
	The communication performance is characterized by (\ref{TightCommunicationExpression}), as shown at the top of the page.
\end{theorem}
\begin{proof}
	Please refer to Appendix F. 
\end{proof}

According to (\ref{TightCommunicationExpression}), the average data rate increases with both the BS density $\lambda_b$ and with the number of resource blocks. We will show in Section \ref{simulations} that the tractable expression given in (\ref{TightCommunicationExpression}) is closely approximated by Monte Carlo simulations. Moreover, if $L = 1$, the interference term $I_1$ can be ignored, and the achievable rate can be simplified to
\vspace{-1.5mm}
\begin{equation}\label{TightCommunicationExpression2}
	\begin{aligned}
		R_c =& \int_0^\infty  { { {\frac{\left(1+p^c z \right)^{({1-M_{\mathrm{t}} })}}{{{z {\rm{H}}_0}\left( {z P^t,\alpha } \right)}}}} }  {\rm{d}}z,
		\vspace{-1.5mm}
	\end{aligned}
\end{equation}
where we have ${\rm{H}}_0( {z,\alpha} ) =  {z^{\frac{2}{\alpha }}}B\left( {\frac{z}{{z + 1}},1 - \frac{2}{\alpha },1 + \frac{2}{\alpha }} \right) + \frac{1}{{{{ {1 + z}}}}}$. In contrast to $L > 1$, the communication rate is independent of the BS density upon dispensing with cooperation, since both the signal power and the interference increase with the BS density. Hence, the SIR remains constant. However, because $R_c$ in Theorem \ref{CommunicationTightExpression} is a complex function of the cooperative size $L$, we seek a more tractable expression for $R_c$ to find the optimal cluster size.

\subsection{Optimal Cooperative Cluster Size}
\label{ApproximatedCommunication}

To maximize the communication performance, we formulate problem (P1) to optimize the cooperative cluster size,
\vspace{-1.5mm}
\begin{alignat}{2}
	\label{P1}
	(\rm{P1}): & \begin{array}{*{20}{c}}
		\mathop {\max }\limits_{L} \quad  R_c
	\end{array} & \\ 
	\mbox{s.t.}\quad
	& p^s + p^c \le P^t, p^s \ge 0, p^c \ge 0, & \tag{\ref{P1}a}\\
	& L, N \ge 1, & \tag{\ref{P1}b} \\
	& R_c + e \times N \le C_{\text{backhaul}}. & \tag{\ref{P1}c}
	\vspace{-1.5mm}
\end{alignat}
To this end, we adopt simplifications for maximizing the expected SIR. First, we simplify the expected data rate as
\vspace{-1.5mm}
\begin{equation}
	\begin{aligned}
		{{\rm{E}}_{r,\Phi _b^S,g_i}}\left[ {\ln \left( {1 + \frac{S}{I}} \right)} \right] &\le {{\rm{E}}_r}\left[ {\ln \left( {1 + {{\rm{E}}_{\Phi _b,g_i}}\left[ {\frac{S}{I}} \right]} \right)} \right] \\
		&\approx {{\rm{E}}_r}\left[ {\ln \left( {1 + \frac{{{{\bar S}_r}}}{{{{\bar I}_r}}}} \right)} \right],
		\vspace{-1.5mm}
	\end{aligned}
\end{equation}
where ${\bar S}_r = {\rm{E}}_{\Phi _b, g_i} [ {g_{1}} + \sum\nolimits_{{{i}} \in {\Phi_a}} {g_{i}} \left\| {\bf{d}}_i \right\|^{-\alpha} r^{\alpha} ] = p^c (M_{\mathrm{t}}-1) ({{r^{ - \alpha }} + \frac{{\pi {\lambda _b}\kappa_c }}{{\alpha  - 2}}( {1 - {\eta _L}^{\alpha  - 2}} ){r^{ - \alpha  + 2}}})$ and ${\bar I}_r = {\rm{E}}_{\Phi _b, g_i} \left[ \sum\nolimits_{{{j}} \in \{\Phi_b \backslash \Phi_a\}} g_j \left\| {\bf{d}}_j \right\|^{-\alpha} r^\alpha  \right] = P^t ( {\frac{{\pi {\lambda _b}}}{{\alpha  - 2}} - \frac{{\pi {\lambda _b}\kappa_c }}{{\alpha  - 2}}( {1 - {\eta _L}^{\alpha  - 2}} )} ){r^{ - \alpha  + 2}}$. Then, it follows that
\vspace{-1.5mm}
\begin{equation}
		{{\rm{E}}_{r,\Phi _b^S,g_i}}\left[ {\ln \left( {1 + \frac{S}{I}} \right)} \right] \approx {{\rm{E}}_r}\left[ {\ln \left( {1 + \frac{{p^c(M_{\mathrm{t}} - 1)}}{P^t} \overline {\rm{SIR}} } \right)} \right],
		\vspace{-1.5mm}
\end{equation}
where we have
\vspace{-1.5mm}
\begin{equation}\label{SIR_Expression}
	\overline {\rm{SIR}} = \frac{{{r^{ - \alpha }} + \frac{{\pi {\lambda _b}\kappa_c }}{{\alpha  - 2}}\left( {1 - {\eta _L}^{\alpha  - 2}} \right){r^{ - \alpha  + 2}}}}{{\left( {\frac{{\pi {\lambda _b}}}{{\alpha  - 2}} - \frac{{\pi {\lambda _b}\kappa_c }}{{\alpha  - 2}}\left( {1 - {\eta _L}^{\alpha  - 2}} \right)} \right){r^{ - \alpha  + 2}}}}.
	\vspace{-1.5mm}
\end{equation}
Therefore, we can optimize the $\overline {\rm{SIR}}$ instead of $R_c$. Observe in (\ref{SIR_Expression}) that the optimal cluster size is independent of the transmission power allocation.
Moreover, in (\ref{SIR_Expression}), the optimal cooperative cluster size is determined by $f(L) = \kappa_c \left( {1 - {\eta _L}^{\alpha  - 2}} \right)$ for any given $r$. If we replace the cooperative cluster size variable $L$ in $\kappa_c$ with the distance ratio $\eta$ between the $L$th closest BS and $r$, the optimal $\eta$ can be uniquely found for any given $r$.
Therefore, we first optimize the optimal distance ratio $\eta^*$ relative to $r$ in the cooperative cluster. After finding the optimal $\eta^*$, the optimal $L^*$ can be acquired.
To this end, we define 
\vspace{-1.5mm}
\begin{equation}
	\tilde \kappa_c = \frac{{\Gamma \left( {\psi ,\mu_c \eta^{-2}} \right)}}{{(\psi  - 1)!}} + \sum\limits_{n = \psi  + 1}^\infty  {\frac{{\left( {\psi  - 1} \right){{\left( {\mu_c \eta^{-2}} \right)}^n}}}{{(n-1) \times n!}}} {e^{ - \mu_c \eta^{-2}}}.
	\vspace{-1.5mm}
\end{equation}
To improve the SIR, we maximize the value $\tilde \kappa_c \left( {1 - {\eta}^{\alpha  - 2}} \right)$. Then the optimal $\eta^*$ can be found. Hence, the approximated optimal cooperative cluster size is given by selecting an $L$, which makes ${\rm{E}}\left[ \frac{r}{r_L} \right] \approx \eta^*$. The simulation results of Section \ref{simulations} will confirm that the cooperative cluster size found closely aligns with the results of the burst search method based on Monte Carlo results.

\section{Tradeoff Between Sensing and Communication}
In Sections \ref{SensingSection} and \ref{CommunicationSection}, we proved that the sensing (communication) performance is monotonically increasing with the size of the cooperative sensing (communication) cluster. However, due to the limited backhaul capacity, there exists a fundamental tradeoff between the S\&C performance at the network level. First, we propose to use the rate-CRLB region (defined further below) to characterize all the achievable communication rate and sensing CRLB pairs under the constraints of the transmit power and the backhaul capacity. Without loss of generality, the rate-CRLB performance region is defined as
\vspace{-1.5mm}
\begin{equation}
	\begin{aligned}
		{\cal{C}}_{c-s}(L,N, p^c, p^s)  = &\bigg\{ ( \hat r_c, {\hat {\rm{crlb}}} ): \hat r_c \le R_c, {\hat {\rm{crlb}}} \ge {\rm{CRLB}}, \\
		 & p^s + p^c \le P^t, R_c + e \times N \le C_\text{backhaul} \bigg\},
		 \vspace{-1.5mm}
	\end{aligned}
\end{equation}
where $(\hat r_c, \hat {\rm{crlb}})$ represents an achievable rate-CRLB performance pair.
A direct way to find the boundary of the rate-CRLB region (as shown in Fig.~\ref{figure2}) is to exhaustively search through the entire of set all feasible variables $(L,N, p^c, p^s)$ and calculate the corresponding S\&C performance expressions derived in Sections \ref{SensingSection} and \ref{CommunicationSection}. However, this operation imposes an excessive computational complexity, especially when the backhaul capacity region is large. 

\begin{figure}[t]
	\centering
	\includegraphics[width=7.2cm]{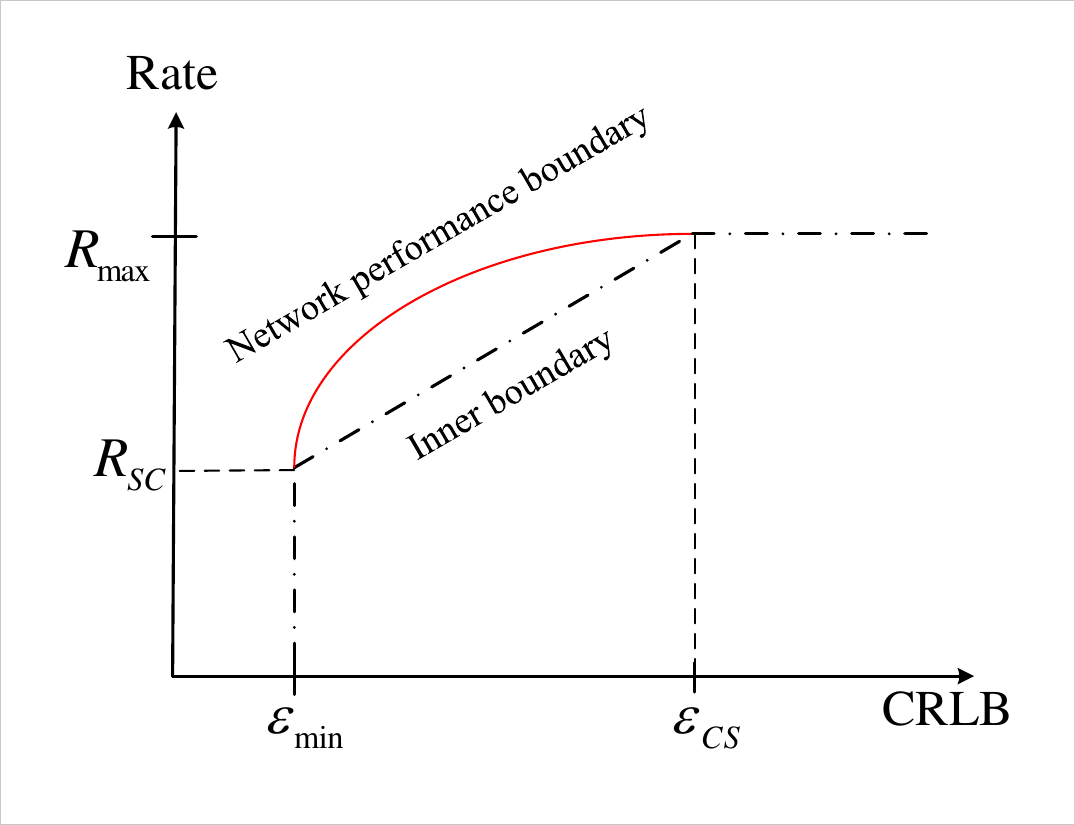}
	\vspace{-4mm}
	\caption{Illustration of the S-C network-level performance region.}
	\label{figure2}
\end{figure}
Indeed, it is sufficient to explore the two dimensions of $L$ and $p^c$, since the optimal values for $N$ and $p^s$ can be uniquely derived for any specific $L$ and $p^c$.
Then, according to the optimal cooperative cluster size of communication-only and sensing-only networks, denoted by $\tilde L^*$ and $\tilde N^*$, the search range can be drastically reduced. Moreover, with any given $L$, the optimal $N$ can be expressed as $N^* = \min(\tilde N^*, \lfloor \frac{C_{\text{backhaul}} - R_c}{e} \rfloor )$. In addition to reducing the search range of $L$ and $N$, we can also decrease the search range of the transmit power according to the characteristics of the boundary. $R_c$ is monotonically increasing with the communication transmit power $p^c$, and $\rm{CRLB}$ is monotonically decreasing with the sensing transmit power $p^s$. Accordingly, the following conclusion may be gleaned for facilitating the problem solution of this.
\begin{Pro}\label{SearchRangeReduce}
	For any given values of $L$ and $p^c$, it corresponds to a rate-CRLB pair $(r'_c, \text{crlb}')$. If there exists another rate-CRLB pair $(r_c, \text{crlb})$ at the current updated boundary of ${\cal{C}}_{c-s}(L,N, p^c, p^s)$ satisfying ${\text{crlb}}' > {\text{crlb}}$ and $r'_c \le r_c$, there is no need to explore the corresponding transmit power range $[P^t-p ^c, \frac{{\text{crlb}}' (P^t-p^c)}{{\text{crlb}}}]$.
\end{Pro}
\begin{proof}
	The proof can be completed by leveraging the monotonic relationship between the rate/CRLB and the transmit power. 
\end{proof}

Based on Proposition \ref{SearchRangeReduce}, the efficiency of the transmit power search can be notably improved.
Next, the second metric of ISAC networks is defined as a function of the data rate and CRLB, given by
\begin{equation}\label{WeightedSum}
	T^{\rm{ISAC}} = \rho \bar R_c + \frac{1-\rho}{ \sqrt{\rm{CRLB}}} ,
\end{equation}
which represents the weighted sum performance of the ISAC network, and $\rho \in [0,1]$. In (\ref{WeightedSum}), $\rho$ represents the weighting factor of the S\&C performance. The S\&C performance of ISAC networks can be flexibly balanced by setting the weighting factor $\rho$ according to the specific requirements.
The problem formulation can be expressed as 
\begin{alignat}{2}
	\label{P2}
	(\rm{P2}): & \begin{array}{*{20}{c}}
		\mathop {\max }\limits_{L,N,p^s,p^c} \quad  T^{\rm{ISAC}}
	\end{array} & \\ 
	\mbox{s.t.}\quad
	& (\ref{P1}a)-(\ref{P1}c). & \nonumber 
\end{alignat}
It is not difficult to prove that the S-C performance boundary is a convex function. Then the optimal total ISAC performance can be obtained by searching for the target C-S performance along the boundary. Moreover, it can be found that (P2) is a monotonic optimization problem and can be optimally solved by the generic Polyblock algorithm of \cite{tuy2000monotonic}.

\section{Simulation Results}
\label{simulations}
Let us now provide some fundamental insights into the characteristics of ISAC networks and validate the accuracy of the tractable expression derived by comparing it to our Monte Carlo simulation results. The numerical simulations are averaged over various networks and realizations of the small-scale channel fading. The system parameters are as follows: the number of transmit antennas is $M_{\mathrm{t}} = 4$, the number of receive antennas is $M_{\mathrm{r}} = 5$, the transmit power is $P_{\mathrm{t}} = 1$W at each BS, the average RCS $\sigma = 1$, the BS density is $\lambda_b = 1/km^2$, $\lambda_u = 1/km^2$, $\lambda_s = 1/km^2$, $\sigma^2_s = - 80$dB, the pathloss coefficients are $\alpha = 4$, $\beta = 2$. Finally, th backhaul capacity is $C_{\text{backhaul}} = 8.6$ bits/s/Hz, and the number of resource blocks is $\psi = 15$.

\subsection{Sensing Performance}
\begin{figure}[t]
	\centering
	\includegraphics[width=7.2cm]{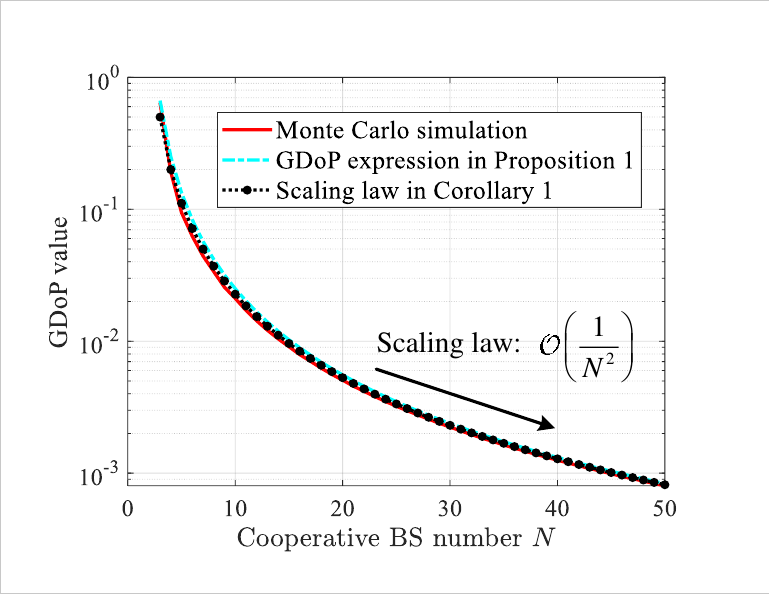}
	\vspace{-3mm}
	\caption{GDoP value vs. the number $N$ of cooperative BSs.}
	\label{figure5}
\end{figure}

To verify the accuracy of our sensing performance analysis, our Monte Carlo simulations are compared to the closed-form expression derived in Section \ref{SensingSection}, as shown in Fig.~\ref{figure5}. Specifically, the disparity between the results outlined in Proposition \ref{GDoPDerivation} and the simulation results is remarkably small.
This demonstrates the effectiveness of the GDoP expression presented in Proposition \ref{GDoPDerivation}.
The scaling law of the GDoP expression derived in Corollary \ref{ScalingLawWithoutDis} is also consistent with the simulation results. Typically, when the size of the cooperative sensing cluster increases from $N = 3$ to $N = 6$, the geometry-based gain increases tenfold.

\begin{figure}[t]
	\centering
	\includegraphics[width=7.2cm]{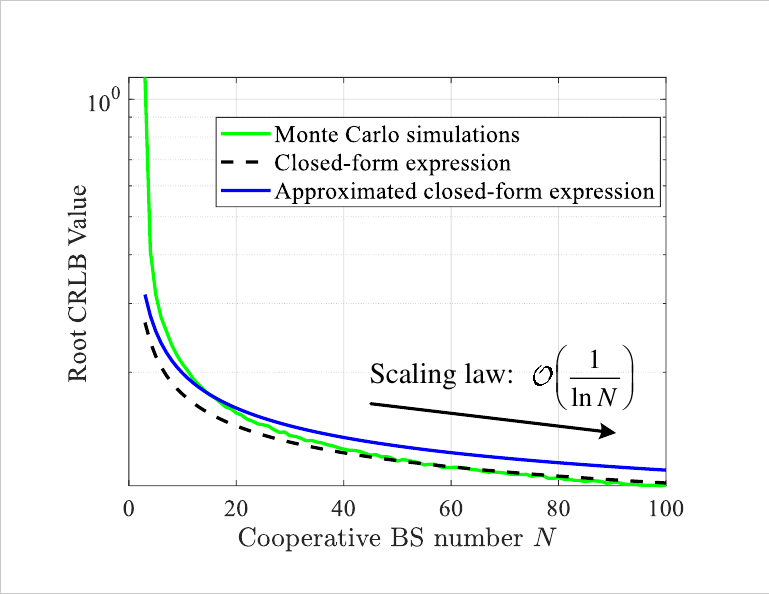}
	\vspace{-3mm}
	\caption{Root CRLB value vs. the number $N$ of cooperative BSs when the configuration of the acceptance probability is 1.}
	\label{figure6}
\end{figure}

In Fig.~\ref{figure6}, the tractable expression derived in Theorem \ref{SimplifiedWithDis3} provides a remarkably tight approximation, especially for a larger number $N$ of cooperative BSs. It is noteworthy that when the number of cooperating BSs is relatively small, for instance, $N \le 4$, the closed-form expressions exhibit a slight deviation from the Monte Carlo simulations. This is mainly due to the less precise calculation of the expectation operation involving trigonometric functions, when the number of ISAC BSs is small.
Furthermore, Fig.~\ref{figure6} reveals that increasing the number of cooperative BSs results in substantial accuracy improvement when the total number of BSs is limited, but it yields only incremental performance gains for $N \ge 10$. This is expected, because the participation of more randomly located BSs in the cooperation leads to increased signal attenuation for the distant BSs, resulting in a performance gain that is significantly lower than that observed for nearby BSs. As seen from Figs.~\ref{figure5} and~\ref{figure6}, the expected CRLB, denoted as ${\rm{tr}}\left(\bar {\bf{F}}^{-1}_N\right)$, exhibits a slower reduction upon increasing $N$ when compared to the GDoP value, ${\rm{tr}}(\tilde {\bf{F}}^{-1}_N)$. This conclusion provides useful insights into strategic ISAC BS deployments, striking a compelling tradeoff between the performance improvement attained and cooperation costs imposed.

Furthermore, Fig.~\ref{figure4} shows that the CRLB decreases first and then increases for $N \ge 15$. The optimal cooperative sensing cluster size equals $\psi = 15$. The main factor is that, as the average number of service requests per BS exceeds the number of allocated resource blocks, each BS is likely to reach its full traffic load. Consequently, some requests sent from nearby targets may be declined, resulting in the forming of a cooperative sensing cluster, where the BSs are situated at a larger distance from the typical target. As the average number of service requests continues to rise, the distance between transceivers and targets also increases.

\begin{figure}[t]
	\centering
	\includegraphics[width=7.2cm]{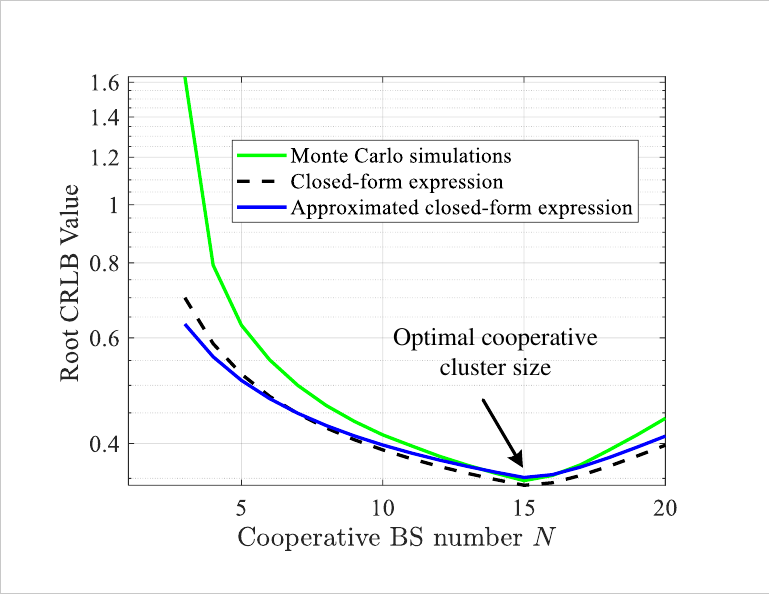}
	\vspace{-3mm}
	\caption{Root CRLB value vs. the number $N$ of cooperative BSs for the cooperation acceptance probability derived in (\ref{AcceptationSensing}).}
	\label{figure4}
\end{figure}

\subsection{Communication Performance}

\begin{figure}[t]
	\centering
	\includegraphics[width=7.2cm]{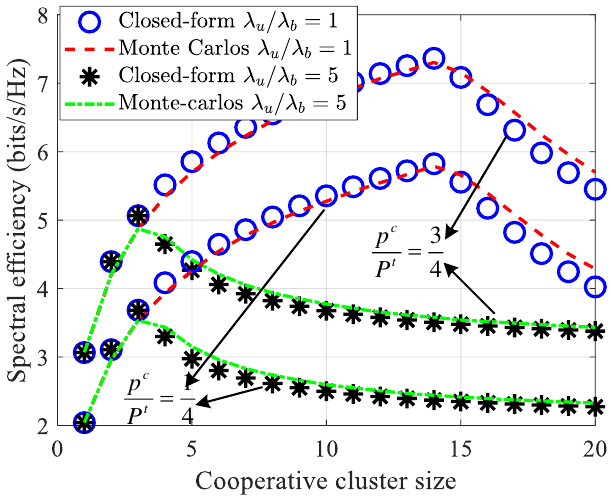}
	\vspace{-3mm}
	\caption{Illustration of cooperative communication performance.}
	\label{figure7}
\end{figure}

Fig.~\ref{figure7} illustrates that the results of the original expression for $R_c$ in Theorem \ref{CommunicationTightExpression} are consistent with the simulation results, which validates the accuracy of our analysis in Section \ref{CommunicationPerformance}. 
As depicted in Fig.~\ref{figure7}, for any given communication power ratio $\frac{p^c}{P^t}$, the communication spectral efficiency $R_c$ initially increases and then decreases with the cooperative cluster size $L$. The primary cause for this trend is the increasing involvement of more users in the service, leading to a reduction in the average acceptance probability for each user. It is evident that regardless of the specific power ratios $\frac{p^c}{P^t}$, there is a consistent optimal value $L^*$ for the same user-BS density $\frac{\lambda_u}{\lambda_b}$ that maximizes the spectral efficiency. This is consistent with the analysis of Section \ref{ApproximatedCommunication}. This optimal value $L^*$ decreases as $\frac{\lambda_u}{\lambda_b}$ increases. This is attributed to the fact that for a higher $\frac{\lambda_u}{\lambda_b}$, more users may send service requests to the BSs, consequently increasing the traffic load of each BS and reducing the acceptance probability.

\begin{figure}[t]
	\centering
	\includegraphics[width=7.2cm]{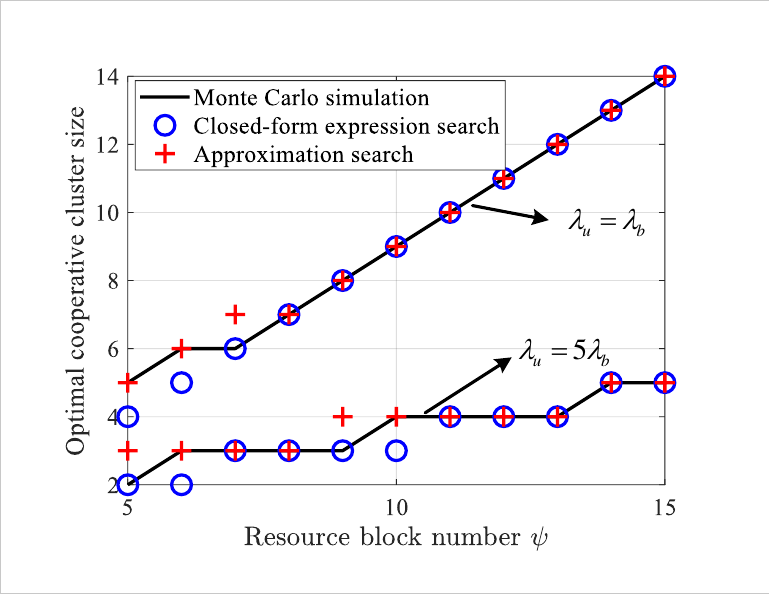}
	\vspace{-3mm}
	\caption{Verification of optimal cooperative cluster size.}
	\label{figure8}
\end{figure}

To verify the effectiveness of the proposed approximated optimal cluster size method (namely approximation search) presented in Section \ref{ApproximatedCommunication} for problem (P1), the optimal cooperative cluster size $L^*$ obtained is compared to that of other exhaustive searches based on Monte Carlo simulation results and to the expression in (\ref{TightCommunicationExpression}), as shown in Fig.~\ref{figure8}. 
It is evident that the optimal cluster size derived from the approximate expression closely aligns with the result obtained through searching based on the Monte Carlo simulation and the expression derived in Theorem {\ref{CommunicationTightExpression}}, especially when the number of resource blocks is larger than 10.
With a fivefold increase in the user-BS density ratio, the optimal cluster size decreases by approximately the same factor, as dictated by the generally satisfied condition of $\frac{\lambda_u \bar L}{\lambda_b} = \psi$. This observation is consistent with the conclusion discussed in Remark \ref{RatioOptimal}.

\subsection{Tradeoff between Sensing and Communication}
\begin{figure}[t]
	\centering
	\includegraphics[width=7.2cm]{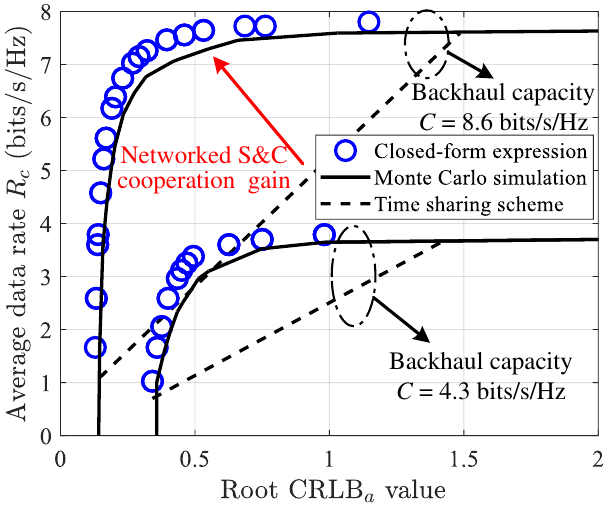}
	\vspace{-3mm}
	\caption{The average data rate vs. root CRLB parametrized by different backhaul capacity constraints.}
	\label{figure9}
\end{figure}
\begin{figure}[t]
	\centering
	\includegraphics[width=7.2cm]{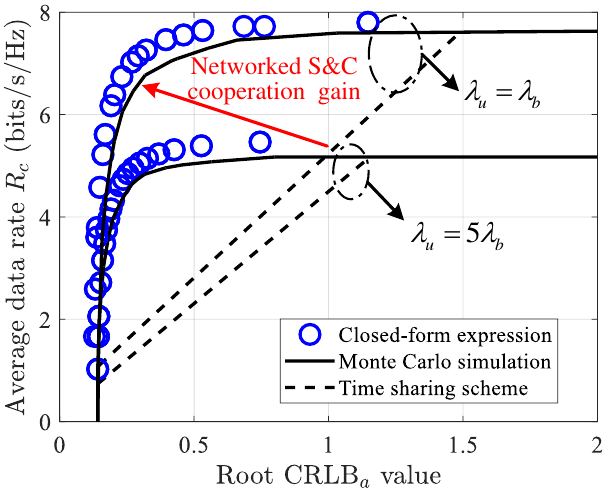}
	\vspace{-3mm}
	\caption{The average data rate vs. root CRLB parametrized by different BS densities.}
	\label{figure10}
\end{figure}

In this subsection, we validate the proposed cooperative ISAC scheme, encompassing both the performance boundary $\mathcal{C}_{\mathrm{C}-\mathrm{S}}$ and the weighted sum performance $T^{\rm{ISAC}}$. First, we compare the effectiveness of the time-sharing scheme based on two corner points for illustrating the performance of the cooperative ISAC scheme under various setups.
The tradeoff profile between the average data rate $R_c$ and the average CRLB is depicted in Fig. \ref{figure9}, confirming both the accuracy of the analytical results and the flexibility of our proposed cooperative ISAC networks. As the backhaul capacity increases, the performance boundaries of S\&C expand significantly. Furthermore, it is observed from Fig. \ref{figure9} that the $( \hat r_c, {\hat {\rm{crlb}}} )$ region of the optimal cooperative scheme becomes much larger than that of the time-sharing scheme, as the backhaul capacity increases. 
For instance, under the backhaul capacity of $C_{\text{backhaul}} = 8.6$ bits/s/Hz and $C_{\text{backhaul}} = 4.3$  bits/s/Hz, the user data rate $R_c$ of the proposed cooperative scheme becomes up to 300\% and 85\% higher than that in the time-sharing scheme, respectively. This outcome stems from the augmented capacity of backhaul links, enabling the network to effectively coordinate the transmit power and multi-cell resources, thereby enhancing the cooperative cluster design gains of S\&C.

Fig.~\ref{figure10} shows that the networked ISAC performance can be substantially extended by exploiting the optimal cooperative strategy under different user-BS densities. Specifically, the proposed cooperative scheme is capable of improving the communication performance by up to 78\% and 50\% as compared to the time-sharing scheme associated with $\frac{\lambda_u}{\lambda_b} = 1$ and $\frac{\lambda_u}{\lambda_b} = 5$, respectively. Similar to the average S\&C performance depicted in Fig.~\ref{figure9}, the rate-CRLB region of the proposed cooperative ISAC scheme significantly expands compared to that of the time-sharing scheme, as the backhaul capacity grows. Additionally, as shown in Fig. \ref{figure10}, as the sensing performance erodes under small root CRLB value, the communication rate improves more significantly. The main reason for this is that when the number of BSs is small, the sensing performance improves rapidly as the number of BSs increases. Fig.~\ref{figure11} shows the optimal weighted joint S\&C performance $T^{\rm{ISAC}}$ for problem (P2) under different weighting factors $\rho$. As the weighting coefficient $\rho$ increases, the sensing performance monotonically decreases, while the communication performance increases. Compared to the non-collaborative scheme, our proposed cooperative ISAC scheme can flexibly enhance both the S\&C performance by appropriately setting the weighting factor.
Additionally, it is observed that $T^{\rm{ISAC}}$ consistently falls below the maximum of the metric ${1}/{\sqrt{{\rm{tr}}({\bf{F}}_N^{-1})}}$ and the communication data rate. This is attributed to the fact that the weighted summation value always resides between the sensing-only and the communication-only performance.

\begin{figure}[t]
	\centering
	\includegraphics[width=7.2cm]{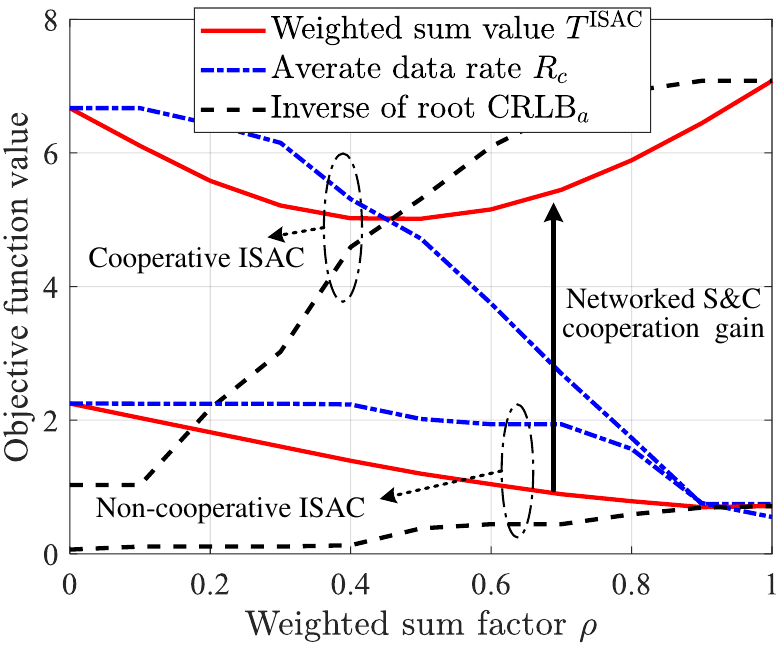}
	\vspace{-3mm}
	\caption{Weighted sum performance comparison versus different factor $\rho$.}
	\label{figure11}
\end{figure}

\section{Conclusion}
In this paper, we proposed a novel cooperative scheme for ISAC networks by simultaneously adopting CoMP-based joint transmission and distributed radar techniques. The S\&C performance expressions were analytically derived using stochastic geometry. We revealed that the average cooperative sensing performance CRLB of the entire ISAC network scales with $\ln^2N$, which is lower than the squared geometry gain of $N^2$. We formulated a profile optimization problem for the ISAC network performance, and compared it to the time-sharing scheme to verify that the optimal cooperative cluster design and power allocation of ISAC networks substantially improves the cooperative gain at the network level. The simulation results demonstrated the benefits of the proposed cooperative ISAC scheme and provided insightful guidelines for designing practical large-scale ISAC networks.

\section*{Appendix A: \textsc{Proof of Proposition \ref{GDoPDerivation}}}
For ease of analysis, the numerator of (\ref{GDoPExpression}) is transformed as follows:
\vspace{-1.5mm}
\begin{equation}\label{numeratorExpression}
	\begin{aligned}
		&\sum\nolimits_{i = 1}^N {\sum\nolimits_{j = 1}^N {a_{ij}^2} }  + \sum\nolimits_{i = 1}^N {\sum\nolimits_{j = 1}^N {b_{ij}^2} }  \\
		=&  2{N^2} + 2\sum\nolimits_{i = 1}^N {\sum\nolimits_{j = 1}^N {\cos \left( {{\theta _i} - {\theta _j}} \right)} }.
		\vspace{-1.5mm}
	\end{aligned}
\end{equation}
Furthermore, the denominator of (\ref{GDoPExpression}) is simplified as
\vspace{-1.5mm}
\begin{equation}\label{denominatorExpression}
	\begin{aligned}
		&\left( {\sum\nolimits_{i = 1}^N {\sum\nolimits_{j = 1}^N {a_{ij}^2} } } \right)\left( {\sum\nolimits_{i = 1}^N {\sum\nolimits_{j = 1}^N {b_{ij}^2} } } \right) \\
		&- {\left( {\sum\nolimits_{i = 1}^N {\sum\nolimits_{j = 1}^N {{a_{ij}}{b_{ij}}} } } \right)^2}  \\
		=& {\sum\nolimits_{l = 1}^N {\sum\nolimits_{k = 1}^N {\sum\nolimits_{i \ge k}^N {\sum\nolimits_{j > {\lceil (k - i)N + l \rceil}^+}^N {{{\left( {{a_{kl}}{b_{ij}} - {a_{ij}}{b_{kl}}} \right)}^2}} } } } }.
		\vspace{-1.5mm}
	\end{aligned}
\end{equation}
By adopting the above transformations, the numerator and denominator become the sum of a series of angle-related variables of the same order. In (\ref{denominatorExpression}), there are $\frac{N^2\left( N-1 \right)}{2}$ items satisfying $k=i$ with the expected value 0.5 of ${{{\left( {{a_{kl}}{b_{ij}} - {a_{ij}}{b_{kl}}} \right)}^2}}$. The trends are also similar for other items, e.g., $k=i$, $l=i$, $l=j$. Then, the expected CRLB can be simplified, so that becomes only related to $N$ as follows:
\vspace{-1.5mm}
\begin{equation}
	\begin{aligned}
		&{\rm{E}}\left[ {{\rm{tr}}\left( {{{\bf{F}}_N^{ - 1}}} \right)} \right] \\
		 = & \frac{{2{N^2} + 2N + 2{\rm{E}}\left[ {\sum\nolimits_{i \ne j}^N {\sum\nolimits_{j = 1}^N {\cos \left( {{\theta _i} - {\theta _j}} \right)} } } \right]}}{{{\rm{E}}\left[ {\sum\nolimits_{l = 1}^N {\sum\nolimits_{k = 1}^N {\sum\nolimits_{i \ge k}^N {\sum\nolimits_{j > {\lceil (k - i)N + l \rceil}^+}^N {{{\left( {{a_{kl}}{b_{ij}} - {a_{ij}}{b_{kl}}} \right)}^2}} } } } } \right]}} \\
		 \overset{(a)}{\approx} & \frac{{2N^2 + 2N}}{{2\frac{N^4-N^2}{2} - \frac{4N^2(N-1)}{4}}} = \frac{{2N + 2}}{{{N^3} - {N^2}}},
		\vspace{-1.5mm}
	\end{aligned}
\end{equation}
where ($a$) holds since we have ${\rm{E}}\left[{\cos \left( {{\theta _i} - {\theta _j}} \right)}\right] = 0$, $\forall i \ne j \in {\cal{N}}$, and ${\rm{E}}\left[ {{{\left( {\sin \left( {\Delta {\theta _{ki}}} \right) + \sin \left( {\Delta {\theta _{kj}}} \right) + \sin \left( {\Delta {\theta _{li}}} \right) + \sin \left( {\Delta {\theta _{lj}}} \right)} \right)}^2}} \right] \approx 2$. This completes the proof.

\section*{Appendix B: \textsc{Proof of Proposition \ref{SimplifiedWithDis1}}}

Firstly, the expected transmit beamforming gain can be formulated as ${\rm{E}}[\left|\mathbf{a}^H(\theta_i) \mathbf{w}_i^s\right|^2] = M_{\rm{t}} - 1$, i.e., $G_t = M_{\rm{t}} - 1$.
Since the BS' location follows a homogeneous PPP, the angle $\theta_i$ and distance $d_i$ are independent for each BS. Thus, it follows that ${{\rm{E}}_{d,\theta }}\left[ {{\rm{tr}}\left( {{{\tilde {\bf{F}}}^{ - 1}}} \right)} \right] = {{\rm{E}}_d}\left[ {{{\rm{E}}_\theta }\left[ {{\rm{tr}}\left( {{{\tilde {\bf{F}}}^{ - 1}}} \right)} \right]} \right]$. Then, the CRLB can be transformed as shown in (\ref{SimplifiedDerivation}), at the top of the next page.
\begin{figure*}
		\begin{align}\label{SimplifiedDerivation}
			{\rm{CRLB}} 
			& \! = \! {{\rm{E}}_d} \!\left[\! {\frac{{2{{\rm{E}}_\theta }\left[ {\sum\nolimits_{i = 1}^N {\sum\nolimits_{j = 1}^N {{D _{ij}}} } } \right] + 2{{\rm{E}}_\theta }\left[ {\sum\nolimits_{i = 1}^N {\sum\nolimits_{j = 1}^N {{D _{ij}}\cos \left( {{\theta _i} - {\theta _j}} \right)} } } \right]}}{{{{\rm{E}}_\theta }\left[ {\sum\nolimits_{l = 1}^N {\sum\nolimits_{k = 1}^N {\sum\nolimits_{i \ge k}^N {\sum\nolimits_{j > \!{\lceil (k - i)N + l \rceil}^+ \!}^N {{D _{kl}}{D _{ij}}} } } } {{\left( {{a_{kl}}{b_{ij}} - {a_{ij}}{b_{kl}}} \right)}^2}} \right]}}} \!\right] \!\!=\! {{\rm{E}}_d}\!\left[ {\frac{{2\sum\nolimits_{i = 1}^N {\sum\nolimits_{j = 1}^N {{D _{ij}}} }  + 2\sum\nolimits_{i = 1}^N {{D _{ii}}} }}{{\!2\sum\nolimits_{l = 1}^N \!{\sum\nolimits_{k = 1}^N \!{\sum\nolimits_{i \ge k}^N \!{\sum\nolimits_{j > \!{\lceil (k - i)N + l \rceil}^+ \!}^N {{D _{kl}}{D _{ij}}} } } } }}} \right] \nonumber \\
			&\approx {{\rm{E}}_d}\left[ {\frac{2}{{\sum\nolimits_{l \ne k}^N {\sum\nolimits_{k = 1}^N {{D _{kl}}} } }}} \right] \approx \frac{2}{{\sum\nolimits_{l \ne k}^N {\sum\nolimits_{k = 1}^N {E{{\left[ {{d_k}} \right]}^{ - \beta}}E{{\left[ {{d_l}} \right]}^{ - \beta}}} } }}.
			\vspace{-3mm}
		\end{align}
	\vspace{-3mm}
	\hrulefill
\end{figure*}
In (\ref{SimplifiedDerivation}), $D_{i,j} = d_i^{-\beta}d_j^{-\beta}$, and the first approximation is adopted by ignoring the items with lower order, and the second approximation holds due to the independent distance of different BSs.

\section*{Appendix C: \textsc{Proof of Theorem \ref{SimplifiedWithDis3}}}
Firstly, when $\beta = 2$, the denominator of (\ref{CRLB_expression}) can be expressed as $\sum\nolimits_{l = 1}^N {\frac{1}{l}\sum\nolimits_{k = 1}^N {\frac{1}{k}} } $.
According to the sum of series, we have ${H_N} = 1 + \frac{1}{2} + \frac{1}{3} +  \ldots . + \frac{1}{N} \approx \ln (N) + \gamma  + \frac{1}{{2N}}$ and $\mathop {\lim }\limits_{N \to \infty } \left[ {{H_N} - \ln (N)} \right] = \gamma ,\gamma  = 0.577$. When $N \to \infty$, it follows that
\vspace{-1.5mm}
	\begin{align}
		&\mathop {\lim }\limits_{N \to \infty } \sum\nolimits_{l = 1}^N {\frac{1}{l}\sum\nolimits_{k = 1}^N {\frac{1}{k}} }  \nonumber \\
		=& \mathop {\lim }\limits_{N \to \infty } \sum\nolimits_{l = 1}^N {\frac{1}{l}\left( {\ln \left( N \right) + \gamma  + \frac{1}{{2N}}} \right)} \nonumber \\
		=& \mathop {\lim }\limits_{N \to \infty } {\left( {\ln \left( N \right) + \gamma  + \frac{1}{{2N}}} \right)^2} = {\ln ^2} N .
		\vspace{-1.5mm}
	\end{align}
This completes the proof.

\section*{Appendix D: \textsc{Proof of Lemma \ref{AcceptationProbability}}}
Under the proposed clustering model, the average association area of each BS is given by ${\cal{A}} = \frac{{{\Gamma ^2}\left( {N + 0.5} \right)}}{{{\lambda _b}{\Gamma ^2}\left( N \right)}}.$
Furthermore, the users are PPP distributed with density $\lambda_s$. Thus the average number of users associated with a BS is ${\cal{A}} \lambda_s$. The probability that a BS accepts the localization service request is equivalent to the probability that the number of users within the BS association area is less than the maximum load, $\psi$. Let ${\cal{N}}\left(|{\cal{A}}| \lambda_u \right)$ be the number of users with density $\lambda_u$ in a geographical area $|{\cal{A}}|$. Furthermore, let $\bar N = \frac{{{\Gamma ^2}\left( {N + 0.5} \right)}}{{{\Gamma ^2}\left( N \right)}}$. Then, $\kappa_s$ can be calculated as follows:
\vspace{-1.5mm}
	\begin{align}
		\kappa_s  = & \sum\limits_{n = 0}^\psi  {\Pr } \left[ {{\cal N}\left( {{\lambda _u} |{{\cal{A}}} | = n} \right)} \right] + \sum\limits_{n = \psi  + 1}^\infty  {\Pr } \left[ {{\cal N}\left( {{\lambda _u} |{{\cal{A}}} | = n} \right)} \right]\frac{\psi }{n} \nonumber \\
		= & \sum\limits_{n = 0}^\psi  {\frac{{{{\left( {\frac{{{\lambda _u}\bar N}}{{{\lambda _b}}}} \right)}^n}}}{{n!}}} {e^{ - \frac{{{\lambda _u}\bar N}}{{{\lambda _b}}}}} + \sum\limits_{n = \psi  + 1}^\infty  {\frac{{{{\left( {\frac{{{\lambda _u}\bar N}}{{{\lambda _b}}}} \right)}^n}}}{{n!}}} {e^{ - \frac{{{\lambda _u}\bar N}}{{{\lambda _b}}}}}\frac{\psi }{n} \nonumber \\
		=& \frac{{\Gamma \left( {\psi ,\mu_s \bar N} \right)}}{{\psi !}} + \sum\limits_{n = \psi  + 1}^\infty  {\frac{{\psi {{\left( {\mu_s \bar N} \right)}^n}}}{{n \times n!}}} {e^{ - \mu_s \bar N}}.
		\vspace{-1.5mm}
	\end{align}
This completes the proof.

\section*{Appendix E: \textsc{Proof of Lemma \ref{LaplaceTransform}}}
The interference term associated with a given distance $r$ from the typical user to the closest BS can be derived by utilizing the Laplace transform. For ease of analysis, we introduce a geometric parameter $\eta_L = \frac{\left\| {{{\bf{d}}_1}} \right\|}{\left\| {{{\bf{d}}_{L}}} \right\|} $, defined as the distance from the closest BS normalized by the distance from the farthest BS in the cluster for the typical user. When $\left\| {{{\bf{d}}_1}} \right\| = r$ and $\left\| {{{\bf{d}}_L}} \right\| = r_L$, we have
\vspace{-1.5mm}
\begin{align}\label{CommunicationEquationExpression}
	&{{\cal L}_{{I_{2}}}}(z)  =  {\rm{E}}_{\Phi_b, g_i} \!\big[ {\exp \big( { - z{{r}^\alpha }\sum\nolimits_{i = L+1}^\infty  \! {{{\left\| {{{\bf{d}}_i}} \right\|}^{ - \alpha }}} {{| {{\bf{h}}_{i}^H{{\bf{W}}_i}} |}^2}} \big)} \! \big] \nonumber  \\
	&\overset{(a)}{=}  {\rm{E}}_{\Phi_b}\!\left[ \!\left( \prod _{{{{\bf{d}}_i}} \in \Phi_b \textbackslash {\cal{O}}(0,r_L)} {   {  {{{{\left( {1 + zP^t{r^\alpha }{{\left\| {{{\bf{d}}_i}} \right\|}^{ - \alpha }}} \right)}^{-1}}}} } dx} \right) \bigg| r, r_L \right]  \nonumber \\
	&\overset{(b)}{=} \exp \left( { - 2\pi \lambda_b \int_{{r_L}}^\infty  {\left( {1 - {{{{\left( {1 + zP^t{r^\alpha }{x^{ - \alpha }}} \right)}^{-1}}}}} \right)} xdx} \right)  \nonumber \\
	&\overset{(c)}{=} \exp \bigg(  - \left. {\pi \lambda_b y\left( {1 - {{{{\left( {1 + zP^t{r^\alpha }{y^{ - \frac{\alpha }{2}}}} \right)}^{-1}}}}} \right)} \right|_{{r_L^2}}^\infty   \nonumber \\
	& - \pi \lambda_b \int_{{r_L^2}}^\infty  {\frac{\alpha }{2}zP^t{r^\alpha }{y^{ - \frac{\alpha }{2}}}{{\left( {1 + z P^t{r^\alpha }{y^{ - \frac{\alpha }{2}}}} \right)}^{ - 2}}} dy \bigg)  \nonumber \\
	&\overset{(d)}{=} \exp \left(  - \pi \lambda_b {r^2}{\rm{H}}_2( {zP^t,\alpha ,\eta_{L} } ) \right) ,
	\vspace{-1.5mm}
\end{align}
where ${\rm{H}}_2\left( {x,\alpha ,\eta_L } \right) = {x^{\frac{2}{\alpha }}}B\left( {\frac{x}{{x + {\eta_L ^{ - \alpha }}}},1 - \frac{2}{\alpha },1 + \frac{2}{\alpha }} \right) + \frac{1}{{{\eta_L ^2}}}\left( {{{{{\left( {1 + x{\eta_L ^\alpha }} \right)}^{-1}}}} - 1} \right)$.
In (\ref{CommunicationEquationExpression}), ($a$) follows from the fact that the small-scale channel fading is independent of the BS locations and that the interference power imposed by each interfering BS at the typical user is distributed as $\Gamma(1,P^t)$. To derive ($b$), we harness the probability generating functional (PGFL) of a PPP with density $\lambda_b$. To elaborate, ($c$) comes from the variable $y = x^2$ and the distribution integral strategies, while ($d$) follows from the distribution integral strategies and $\eta_{L} = \frac{r}{r_L}$.  
Similarly, the Laplace transform of useful signals can be expressed by 
\vspace{-1.5mm}
\begin{equation}
	{\rm{E}}\!\left[ {{e^{ - z U}}} \right] \!=\! \exp \!\bigg( \! - \pi \kappa_c \lambda_b {r^2}{\rm{H}}_1\left( { zp^c,M_{\mathrm{t}}-1,\alpha ,\eta_{L} } \right) \!\bigg),
	\vspace{-1.5mm}
\end{equation}
where ${\rm{H}}_1\left( {x,K,\alpha ,\eta_L } \right)  = \frac{1}{{{\eta ^2}}}\left( {1 - \frac{1}{{{{\left( {1 + x{\eta ^\alpha }} \right)}^K}}}} \right) + \frac{1}{{{{\left( {1 + x} \right)}^K}}} - 1 + K{x^{\frac{2}{\alpha }}}\! 
\left(\! {B\left( {\frac{x}{{x + 1}},1 - \frac{2}{\alpha },K + \frac{2}{\alpha }} \right) \! - \! B\left( {\frac{{x{\eta ^\alpha }}}{{x{\eta ^\alpha } + 1}},1 - \frac{2}{\alpha }, K + \frac{2}{\alpha }} \right)} \right)$.
Similarly, for the intra-cluster interference, i.e., for the BSs declining a service request, the Laplace transform $I_1$ can be formulated by
\vspace{-1.5mm}
\begin{equation}\label{CommunicationEquationExpression2}
	\begin{aligned}
		{{\cal L}_{{I_{1}}}}(z) &= {\rm{E}}_{\Phi_b, g_i}\left[ {\exp \left( { - z P^t{{\left\| {{{\bf{d}}_1}} \right\|}^\alpha }\sum\nolimits_{i \in \{\Phi_c \backslash \Phi_a \}}  {{{\left\| {{{\bf{d}}_i}} \right\|}^{ - \alpha }}} {g_j}} \right)} \right]  \\
		&{=} \exp \bigg(  - \pi (1-\kappa_c) \lambda_b {r^2}{\rm{H}}_1\left( {z P^t,1,\alpha ,\eta_{L} } \right) \bigg) .
		\vspace{-1.5mm}
	\end{aligned}
\end{equation}
This completes the proof.

\section*{Appendix F: \textsc{Proof of Theorem \ref{CommunicationTightExpression}}}

\begin{figure*}
	\begin{equation}\label{ExpressionTransformation}
		\begin{aligned}
			{\rm{E}}\left[ {\log \left( {1 + {\rm{SI}}{{\rm{R}}_c}} \right)\mid r} \right] 
			= \int_0^\infty  {\frac{{1 - {\rm{E}}\left[ {{e^{ - z{P_c}{g_1}}}} \right]{\rm{E}}\left[ {{e^{ - z U }}} \right]}}{z}} {\rm{E}}\left[ {{e^{ - z I_1 }}} \right]{\rm{E}}\left[ {{e^{ - z I_2 }}} \right]{\rm{d}}z. 
			\vspace{-3mm}
		\end{aligned}
	\end{equation}
	\vspace{-3mm}
	\hrulefill
\end{figure*}
According to (\ref{CommunicationBasicEquation}), the data rate under the conditional distance $r$ is given by (\ref{ExpressionTransformation}), as shown at the top of the next page.
Then, the conditional expected spectrum efficiency is given by
\vspace{-1.5mm}
\begin{equation}\label{longEquation2}
	\begin{aligned}
		&\int_0^\infty \!\! \int_0^1 \!\! \int_0^\infty \!\! \frac{{1 \! - \! {\rm{E}}[ {{e^{ - z{g_1}}}} ]\exp \left( { - \pi \kappa_c \lambda_b {r^2}{{\rm{H}}_1}\!( { zp^c,M_{\mathrm{t}}\! - \! 1,\alpha ,\eta_{L} } )} \right)}}{z}  \\
		&\times \exp \left( { - \pi \left( {1 - \kappa_c } \right)\lambda_b {r^2}{{\rm{H}}_1}\left( {z P^t,1,\alpha ,\eta_{L} } \right)} \right)\exp \big(  - \pi \lambda_b {r^2}  \\
		& {{\rm{H}}_2}\left( {zP^t,\alpha ,\eta_{L} } \right)\big) {f_{\eta_L}}\left( \eta  \right) {f_r}\left( r \right) {\rm{d}} r  {\rm{d}}\eta  {\rm{d}}z.
		\vspace{-1.5mm}
	\end{aligned}
\end{equation}
where we have ${f_r}\left( r \right) = 2\pi {\lambda _b}r{e^{ - \pi {\lambda _b}{r^2}}}$.
According to Lemma 3 in \cite{Zhang2014StochasticGeometry}, the PDF of the distance ratio $\eta_L$ can be given by 
\vspace{-1.5mm}
\begin{equation}\label{RatioEquation}
	{f_{\eta_L}}\left( x \right) = 2\left( {L - 1} \right)x{\left( {1 - {x^2}} \right)^{L - 2}}.
	\vspace{-1.5mm}
\end{equation}
Then, by substituting the Laplace transforms of useful signal and interference in Lemma \ref{LaplaceTransform} into (\ref{ExpressionTransformation}), equation (\ref{CommunicationTightExpression}) can be obtained. 
This completes the proof.

\vspace{-1.5mm}
\footnotesize  	
\bibliography{mybibfile}
\bibliographystyle{IEEEtran}

\end{document}